\documentclass[journal,twoside]{IEEEtran} 

\usepackage{cs_techrpt_cover_crest}

\usepackage{mathptmx}

\usepackage[nocompress]{cite}
\usepackage[latin1]{inputenc} 
\usepackage{graphicx}
\usepackage[tight,TABTOPCAP]{subfigure}
\usepackage{amsmath}
\usepackage{amssymb}
\usepackage[noend]{algorithmic}
\usepackage[section]{algorithm}
\usepackage{color}
\usepackage{xspace}
\usepackage{verbatim}

\usepackage{floatflt}

\usepackage{ifpdf} 
\usepackage{wrapfig}

\makeatletter
\def\section{\@startsection {section}{1}{\z@}{-3.25ex plus -1ex minus
 -.2ex}{1.5ex plus .2ex}{\large\bf}}
\def\subsection{\@startsection{subsection}{2}{\z@}{-3.25ex plus -1ex minus
 -.2ex}{1.5ex plus .2ex}{\normalsize\bf}}
\def\subsubsection{\@startsection{subsubsection}{3}{\z@}{3.25ex plus
 1ex minus .2ex}{-1em}{\normalsize\bf}}
\def\paragraph{\@startsection{paragraph}{4}{\z@}{3.25ex plus 1ex minus
  .2ex}{-1em}{\normalsize\bf}}
\def\subparagraph{\@startsection{subparagraph}{4}{\parindent}{3.25ex
  plus 1ex minus .2ex}{-1em}{\normalsize\bf}}
\makeatother

\newcommand{\verbfont}%
{%
    \fontsize{7pt}{7.0pt}%
}
\newcounter{verbline}
\makeatletter%
\def\verblineinput#1{
\begingroup
\setcounter{verbline}{0}%
\everypar{{\makebox[0.4cm][r]{\addtocounter{verbline}{1}~}}}%
\verbfont%
\def\par{\leavevmode\null\@@par}%
\obeylines%
\frenchspacing\@vobeyspaces%
\@makeother\$\@makeother\&\@makeother\#%
\@makeother\^\@makeother\_\@makeother\~%
\@makeother\{\@makeother\}%
\@makeother\%\tt\@noligs\input{#1}\endgroup} \makeatother

\newenvironment{smallerfont}{\fontsize{8}{10}\selectfont}{\normalsize}

 \newcommand{\smallsec}[1]{\smallskip\noindent{\bf #1.}}

\newcommand{\blast}{\textsc{Blast}\xspace}
\newcommand{\slam}{\textsc{SLAM}\xspace}

\newcommand{\cil}{{\sc Cil}}
\newcommand{\cpachecker}{\textsc{\small CPAchecker}\xspace}
\newcommand{\mathsat}{\textsc{MathSAT}\xspace}

\newcommand{\satabs}{\textsc{SATabs}\xspace}
\newcommand{\calysto}{\textsc{Calysto}\xspace}

\newcommand{\true}{\mathit{true}}

\newcommand{\seq}[1]{{\langle #1 \rangle}}

\newcommand{\locs}{\mathit{L}}
\newcommand{\op}{\mathit{op}}
\newcommand{\pc}{\mathit{l}}
\newcommand{\pci}{{\pc_0}}
\newcommand{\pct}{{\pc_E}}

\newcommand{\preds}{\mathcal{P}}

\newcommand{\concr}{\ensuremath{{\cal C}}}
\newcommand{\cart}{\mathbb{C}}
\newcommand{\bool}{\mathbb{B}}

\newcommand{\Ints}{\mathbb{Z}}

\renewcommand{\implies}{\Rightarrow}

\renewcommand{\prec}{\pi}
\newcommand{\PREC}{\Pi}
\newcommand{\path}{\sigma}

\renewcommand{\phi}{\varphi}
\newcommand{\SP}{{\sf SP}}

\newcommand{\sequence}{\ensuremath{\,;\,}}
\newcommand{\abs}[1]{\ensuremath{{#1}}}

\renewcommand{\iff}{\Leftrightarrow}

\newenvironment{example}
    {
      \medskip
      \noindent
      {{\bf Example.}}
    }
    { 
      \hfill \qed
      \medskip
    }

\def\qed{\relax\ifmmode\hfill \Box\else\unskip\nobreak\hfill $\Box$\fi}

\newtheorem{thm}{Theorem}[section]
\newtheorem{cor}[thm]{Corollary}
\newtheorem{prop}[thm]{Proposition}
\newtheorem{lem}[thm]{Lemma}

\newenvironment{theorem}{\vspace{-\lastskip}\par\addvspace{.6pc plus
    .2pc minus .1pc}\begin{thm}}{\end{thm}\par\addvspace{.6pc plus
    .2pc minus .1pc}} \newenvironment{lemma}{\vspace{-\lastskip}\par
  \addvspace{.6pc plus .2pc minus .1pc}\begin{lem}}
  {\end{lem}\par\addvspace{.6pc plus .2pc minus .1pc}}

\newcommand{\ignore}[1]{}

\definecolor {snow}                {rgb}{1.00,0.98,0.98}
\definecolor {ghostwhite}          {rgb}{0.97,0.97,1.00}
\definecolor {whitesmoke}          {rgb}{0.96,0.96,0.96}
\definecolor {gainsboro}           {rgb}{0.86,0.86,0.86}
\definecolor {floralwhite}         {rgb}{1.00,0.98,0.94}
\definecolor {oldlace}             {rgb}{0.99,0.96,0.90}
\definecolor {linen}               {rgb}{0.98,0.94,0.90}
\definecolor {antiquewhite}        {rgb}{0.98,0.92,0.84}
\definecolor {papayawhip}          {rgb}{1.00,0.94,0.84}
\definecolor {blanchedalmond}      {rgb}{1.00,0.92,0.80}
\definecolor {bisque}              {rgb}{1.00,0.89,0.77}
\definecolor {peachpuff}           {rgb}{1.00,0.85,0.73}
\definecolor {navajowhite}         {rgb}{1.00,0.87,0.68}
\definecolor {moccasin}            {rgb}{1.00,0.89,0.71}
\definecolor {cornsilk}            {rgb}{1.00,0.97,0.86}
\definecolor {ivory}               {rgb}{1.00,1.00,0.94}
\definecolor {lemonchiffon}        {rgb}{1.00,0.98,0.80}
\definecolor {seashell}            {rgb}{1.00,0.96,0.93}
\definecolor {honeydew}            {rgb}{0.94,1.00,0.94}
\definecolor {mintcream}           {rgb}{0.96,1.00,0.98}
\definecolor {azure}               {rgb}{0.94,1.00,1.00}
\definecolor {aliceblue}           {rgb}{0.94,0.97,1.00}
\definecolor {lavender}            {rgb}{0.90,0.90,0.98}
\definecolor {lavenderblush}       {rgb}{1.00,0.94,0.96}
\definecolor {mistyrose}           {rgb}{1.00,0.89,0.88}
\definecolor {white}               {rgb}{1.00,1.00,1.00}
\definecolor {black}               {rgb}{0.00,0.00,0.00}
\definecolor {darkslategray}       {rgb}{0.18,0.31,0.31}
\definecolor {dimgray}             {rgb}{0.41,0.41,0.41}
\definecolor {slategray}           {rgb}{0.44,0.50,0.56}
\definecolor {lightslategray}      {rgb}{0.47,0.53,0.60}
\definecolor {gray}                {rgb}{0.75,0.75,0.75}
\definecolor {lightgrey}           {rgb}{0.83,0.83,0.83}
\definecolor {midnightblue}        {rgb}{0.10,0.10,0.44}
\definecolor {navy}                {rgb}{0.00,0.00,0.50}
\definecolor {cornflowerblue}      {rgb}{0.39,0.58,0.93}
\definecolor {darkslateblue}       {rgb}{0.28,0.24,0.55}
\definecolor {slateblue}           {rgb}{0.42,0.35,0.80}
\definecolor {mediumslateblue}     {rgb}{0.48,0.41,0.93}
\definecolor {lightslateblue}      {rgb}{0.52,0.44,1.00}
\definecolor {mediumblue}          {rgb}{0.00,0.00,0.80}
\definecolor {royalblue}           {rgb}{0.25,0.41,0.88}
\definecolor {blue}                {rgb}{0.00,0.00,1.00}
\definecolor {dodgerblue}          {rgb}{0.12,0.56,1.00}
\definecolor {deepskyblue}         {rgb}{0.00,0.75,1.00}
\definecolor {skyblue}             {rgb}{0.53,0.81,0.92}
\definecolor {lightskyblue}        {rgb}{0.53,0.81,0.98}
\definecolor {steelblue}           {rgb}{0.27,0.51,0.71}
\definecolor {lightsteelblue}      {rgb}{0.69,0.77,0.87}
\definecolor {lightblue}           {rgb}{0.68,0.85,0.90}
\definecolor {powderblue}          {rgb}{0.69,0.88,0.90}
\definecolor {paleturquoise}       {rgb}{0.69,0.93,0.93}
\definecolor {darkturquoise}       {rgb}{0.00,0.81,0.82}
\definecolor {mediumturquoise}     {rgb}{0.28,0.82,0.80}
\definecolor {turquoise}           {rgb}{0.25,0.88,0.82}
\definecolor {cyan}                {rgb}{0.00,1.00,1.00}
\definecolor {lightcyan}           {rgb}{0.88,1.00,1.00}
\definecolor {cadetblue}           {rgb}{0.37,0.62,0.63}
\definecolor {mediumaquamarine}    {rgb}{0.40,0.80,0.67}
\definecolor {aquamarine}          {rgb}{0.50,1.00,0.83}
\definecolor {darkgreen}           {rgb}{0.00,0.39,0.00}
\definecolor {darkolivegreen}      {rgb}{0.33,0.42,0.18}
\definecolor {darkseagreen}        {rgb}{0.56,0.74,0.56}
\definecolor {seagreen}            {rgb}{0.18,0.55,0.34}
\definecolor {mediumseagreen}      {rgb}{0.24,0.70,0.44}
\definecolor {lightseagreen}       {rgb}{0.13,0.70,0.67}
\definecolor {palegreen}           {rgb}{0.60,0.98,0.60}
\definecolor {springgreen}         {rgb}{0.00,1.00,0.50}
\definecolor {lawngreen}           {rgb}{0.49,0.99,0.00}
\definecolor {green}               {rgb}{0.00,1.00,0.00}
\definecolor {chartreuse}          {rgb}{0.50,1.00,0.00}
\definecolor {mediumspringgreen}   {rgb}{0.00,0.98,0.60}
\definecolor {greenyellow}         {rgb}{0.68,1.00,0.18}
\definecolor {limegreen}           {rgb}{0.20,0.80,0.20}
\definecolor {yellowgreen}         {rgb}{0.60,0.80,0.20}
\definecolor {forestgreen}         {rgb}{0.13,0.55,0.13}
\definecolor {olivedrab}           {rgb}{0.42,0.56,0.14}
\definecolor {darkkhaki}           {rgb}{0.74,0.72,0.42}
\definecolor {khaki}               {rgb}{0.94,0.90,0.55}
\definecolor {palegoldenrod}       {rgb}{0.93,0.91,0.67}
\definecolor {lightgoldenrodyellow} {rgb}{0.98,0.98,0.82}
\definecolor {lightyellow}         {rgb}{1.00,1.00,0.88}
\definecolor {yellow}              {rgb}{1.00,1.00,0.00}
\definecolor {gold}                {rgb}{1.00,0.84,0.00}
\definecolor {lightgoldenrod}      {rgb}{0.93,0.87,0.51}
\definecolor {goldenrod}           {rgb}{0.85,0.65,0.13}
\definecolor {darkgoldenrod}       {rgb}{0.72,0.53,0.04}
\definecolor {rosybrown}           {rgb}{0.74,0.56,0.56}
\definecolor {indianred}           {rgb}{0.80,0.36,0.36}
\definecolor {saddlebrown}         {rgb}{0.55,0.27,0.07}
\definecolor {sienna}              {rgb}{0.63,0.32,0.18}
\definecolor {peru}                {rgb}{0.80,0.52,0.25}
\definecolor {burlywood}           {rgb}{0.87,0.72,0.53}
\definecolor {beige}               {rgb}{0.96,0.96,0.86}
\definecolor {wheat}               {rgb}{0.96,0.87,0.70}
\definecolor {sandybrown}          {rgb}{0.96,0.64,0.38}
\definecolor {tan}                 {rgb}{0.82,0.71,0.55}
\definecolor {chocolate}           {rgb}{0.82,0.41,0.12}
\definecolor {firebrick}           {rgb}{0.70,0.13,0.13}
\definecolor {brown}               {rgb}{0.65,0.16,0.16}
\definecolor {darksalmon}          {rgb}{0.91,0.59,0.48}
\definecolor {salmon}              {rgb}{0.98,0.50,0.45}
\definecolor {lightsalmon}         {rgb}{1.00,0.63,0.48}
\definecolor {orange}              {rgb}{1.00,0.65,0.00}
\definecolor {darkorange}          {rgb}{1.00,0.55,0.00}
\definecolor {coral}               {rgb}{1.00,0.50,0.31}
\definecolor {lightcoral}          {rgb}{0.94,0.50,0.50}
\definecolor {tomato}              {rgb}{1.00,0.39,0.28}
\definecolor {orangered}           {rgb}{1.00,0.27,0.00}
\definecolor {red}                 {rgb}{1.00,0.00,0.00}
\definecolor {hotpink}             {rgb}{1.00,0.41,0.71}
\definecolor {deeppink}            {rgb}{1.00,0.08,0.58}
\definecolor {pink}                {rgb}{1.00,0.75,0.80}
\definecolor {lightpink}           {rgb}{1.00,0.71,0.76}
\definecolor {palevioletred}       {rgb}{0.86,0.44,0.58}
\definecolor {maroon}              {rgb}{0.69,0.19,0.38}
\definecolor {mediumvioletred}     {rgb}{0.78,0.08,0.52}
\definecolor {violetred}           {rgb}{0.82,0.13,0.56}
\definecolor {magenta}             {rgb}{1.00,0.00,1.00}
\definecolor {violet}              {rgb}{0.93,0.51,0.93}
\definecolor {plum}                {rgb}{0.87,0.63,0.87}
\definecolor {orchid}              {rgb}{0.85,0.44,0.84}
\definecolor {mediumorchid}        {rgb}{0.73,0.33,0.83}
\definecolor {darkorchid}          {rgb}{0.60,0.20,0.80}
\definecolor {darkviolet}          {rgb}{0.58,0.00,0.83}
\definecolor {blueviolet}          {rgb}{0.54,0.17,0.89}
\definecolor {purple}              {rgb}{0.63,0.13,0.94}
\definecolor {mediumpurple}        {rgb}{0.58,0.44,0.86}
\definecolor {thistle}             {rgb}{0.85,0.75,0.85}
\definecolor {snow2}               {rgb}{0.93,0.91,0.91}
\definecolor {snow3}               {rgb}{0.80,0.79,0.79}
\definecolor {snow4}               {rgb}{0.55,0.54,0.54}
\definecolor {seashell2}           {rgb}{0.93,0.90,0.87}
\definecolor {seashell3}           {rgb}{0.80,0.77,0.75}
\definecolor {seashell4}           {rgb}{0.55,0.53,0.51}
\definecolor {antiquewhite1}       {rgb}{1.00,0.94,0.86}
\definecolor {antiquewhite2}       {rgb}{0.93,0.87,0.80}
\definecolor {antiquewhite3}       {rgb}{0.80,0.75,0.69}
\definecolor {antiquewhite4}       {rgb}{0.55,0.51,0.47}
\definecolor {bisque2}             {rgb}{0.93,0.84,0.72}
\definecolor {bisque3}             {rgb}{0.80,0.72,0.62}
\definecolor {bisque4}             {rgb}{0.55,0.49,0.42}
\definecolor {peachpuff2}          {rgb}{0.93,0.80,0.68}
\definecolor {peachpuff3}          {rgb}{0.80,0.69,0.58}
\definecolor {peachpuff4}          {rgb}{0.55,0.47,0.40}
\definecolor {navajowhite2}        {rgb}{0.93,0.81,0.63}
\definecolor {navajowhite3}        {rgb}{0.80,0.70,0.55}
\definecolor {navajowhite4}        {rgb}{0.55,0.47,0.37}
\definecolor {lemonchiffon2}       {rgb}{0.93,0.91,0.75}
\definecolor {lemonchiffon3}       {rgb}{0.80,0.79,0.65}
\definecolor {lemonchiffon4}       {rgb}{0.55,0.54,0.44}
\definecolor {cornsilk2}           {rgb}{0.93,0.91,0.80}
\definecolor {cornsilk3}           {rgb}{0.80,0.78,0.69}
\definecolor {cornsilk4}           {rgb}{0.55,0.53,0.47}
\definecolor {ivory2}              {rgb}{0.93,0.93,0.88}
\definecolor {ivory3}              {rgb}{0.80,0.80,0.76}
\definecolor {ivory4}              {rgb}{0.55,0.55,0.51}
\definecolor {honeydew2}           {rgb}{0.88,0.93,0.88}
\definecolor {honeydew3}           {rgb}{0.76,0.80,0.76}
\definecolor {honeydew4}           {rgb}{0.51,0.55,0.51}
\definecolor {lavenderblush2}      {rgb}{0.93,0.88,0.90}
\definecolor {lavenderblush3}      {rgb}{0.80,0.76,0.77}
\definecolor {lavenderblush4}      {rgb}{0.55,0.51,0.53}
\definecolor {mistyrose2}          {rgb}{0.93,0.84,0.82}
\definecolor {mistyrose3}          {rgb}{0.80,0.72,0.71}
\definecolor {mistyrose4}          {rgb}{0.55,0.49,0.48}
\definecolor {azure2}              {rgb}{0.88,0.93,0.93}
\definecolor {azure3}              {rgb}{0.76,0.80,0.80}
\definecolor {azure4}              {rgb}{0.51,0.55,0.55}
\definecolor {slateblue1}          {rgb}{0.51,0.44,1.00}
\definecolor {slateblue2}          {rgb}{0.48,0.40,0.93}
\definecolor {slateblue3}          {rgb}{0.41,0.35,0.80}
\definecolor {slateblue4}          {rgb}{0.28,0.24,0.55}
\definecolor {royalblue1}          {rgb}{0.28,0.46,1.00}
\definecolor {royalblue2}          {rgb}{0.26,0.43,0.93}
\definecolor {royalblue3}          {rgb}{0.23,0.37,0.80}
\definecolor {royalblue4}          {rgb}{0.15,0.25,0.55}
\definecolor {blue2}               {rgb}{0.00,0.00,0.93}
\definecolor {blue4}               {rgb}{0.00,0.00,0.55}
\definecolor {dodgerblue2}         {rgb}{0.11,0.53,0.93}
\definecolor {dodgerblue3}         {rgb}{0.09,0.45,0.80}
\definecolor {dodgerblue4}         {rgb}{0.06,0.31,0.55}
\definecolor {steelblue1}          {rgb}{0.39,0.72,1.00}
\definecolor {steelblue2}          {rgb}{0.36,0.67,0.93}
\definecolor {steelblue3}          {rgb}{0.31,0.58,0.80}
\definecolor {steelblue4}          {rgb}{0.21,0.39,0.55}
\definecolor {deepskyblue2}        {rgb}{0.00,0.70,0.93}
\definecolor {deepskyblue3}        {rgb}{0.00,0.60,0.80}
\definecolor {deepskyblue4}        {rgb}{0.00,0.41,0.55}
\definecolor {skyblue1}            {rgb}{0.53,0.81,1.00}
\definecolor {skyblue2}            {rgb}{0.49,0.75,0.93}
\definecolor {skyblue3}            {rgb}{0.42,0.65,0.80}
\definecolor {skyblue4}            {rgb}{0.29,0.44,0.55}
\definecolor {lightskyblue1}       {rgb}{0.69,0.89,1.00}
\definecolor {lightskyblue2}       {rgb}{0.64,0.83,0.93}
\definecolor {lightskyblue3}       {rgb}{0.55,0.71,0.80}
\definecolor {lightskyblue4}       {rgb}{0.38,0.48,0.55}
\definecolor {slategray1}          {rgb}{0.78,0.89,1.00}
\definecolor {slategray2}          {rgb}{0.73,0.83,0.93}
\definecolor {slategray3}          {rgb}{0.62,0.71,0.80}
\definecolor {slategray4}          {rgb}{0.42,0.48,0.55}
\definecolor {lightsteelblue1}     {rgb}{0.79,0.88,1.00}
\definecolor {lightsteelblue2}     {rgb}{0.74,0.82,0.93}
\definecolor {lightsteelblue3}     {rgb}{0.64,0.71,0.80}
\definecolor {lightsteelblue4}     {rgb}{0.43,0.48,0.55}
\definecolor {lightblue1}          {rgb}{0.75,0.94,1.00}
\definecolor {lightblue2}          {rgb}{0.70,0.87,0.93}
\definecolor {lightblue3}          {rgb}{0.60,0.75,0.80}
\definecolor {lightblue4}          {rgb}{0.41,0.51,0.55}
\definecolor {lightcyan2}          {rgb}{0.82,0.93,0.93}
\definecolor {lightcyan3}          {rgb}{0.71,0.80,0.80}
\definecolor {lightcyan4}          {rgb}{0.48,0.55,0.55}
\definecolor {paleturquoise1}      {rgb}{0.73,1.00,1.00}
\definecolor {paleturquoise2}      {rgb}{0.68,0.93,0.93}
\definecolor {paleturquoise3}      {rgb}{0.59,0.80,0.80}
\definecolor {paleturquoise4}      {rgb}{0.40,0.55,0.55}
\definecolor {cadetblue1}          {rgb}{0.60,0.96,1.00}
\definecolor {cadetblue2}          {rgb}{0.56,0.90,0.93}
\definecolor {cadetblue3}          {rgb}{0.48,0.77,0.80}
\definecolor {cadetblue4}          {rgb}{0.33,0.53,0.55}
\definecolor {turquoise1}          {rgb}{0.00,0.96,1.00}
\definecolor {turquoise2}          {rgb}{0.00,0.90,0.93}
\definecolor {turquoise3}          {rgb}{0.00,0.77,0.80}
\definecolor {turquoise4}          {rgb}{0.00,0.53,0.55}
\definecolor {cyan2}               {rgb}{0.00,0.93,0.93}
\definecolor {cyan3}               {rgb}{0.00,0.80,0.80}
\definecolor {cyan4}               {rgb}{0.00,0.55,0.55}
\definecolor {darkslategray1}      {rgb}{0.59,1.00,1.00}
\definecolor {darkslategray2}      {rgb}{0.55,0.93,0.93}
\definecolor {darkslategray3}      {rgb}{0.47,0.80,0.80}
\definecolor {darkslategray4}      {rgb}{0.32,0.55,0.55}
\definecolor {aquamarine2}         {rgb}{0.46,0.93,0.78}
\definecolor {aquamarine4}         {rgb}{0.27,0.55,0.45}
\definecolor {darkseagreen1}       {rgb}{0.76,1.00,0.76}
\definecolor {darkseagreen2}       {rgb}{0.71,0.93,0.71}
\definecolor {darkseagreen3}       {rgb}{0.61,0.80,0.61}
\definecolor {darkseagreen4}       {rgb}{0.41,0.55,0.41}
\definecolor {seagreen1}           {rgb}{0.33,1.00,0.62}
\definecolor {seagreen2}           {rgb}{0.31,0.93,0.58}
\definecolor {seagreen3}           {rgb}{0.26,0.80,0.50}
\definecolor {palegreen1}          {rgb}{0.60,1.00,0.60}
\definecolor {palegreen2}          {rgb}{0.56,0.93,0.56}
\definecolor {palegreen3}          {rgb}{0.49,0.80,0.49}
\definecolor {palegreen4}          {rgb}{0.33,0.55,0.33}
\definecolor {springgreen2}        {rgb}{0.00,0.93,0.46}
\definecolor {springgreen3}        {rgb}{0.00,0.80,0.40}
\definecolor {springgreen4}        {rgb}{0.00,0.55,0.27}
\definecolor {green2}              {rgb}{0.00,0.93,0.00}
\definecolor {green3}              {rgb}{0.00,0.80,0.00}
\definecolor {green4}              {rgb}{0.00,0.55,0.00}
\definecolor {chartreuse2}         {rgb}{0.46,0.93,0.00}
\definecolor {chartreuse3}         {rgb}{0.40,0.80,0.00}
\definecolor {chartreuse4}         {rgb}{0.27,0.55,0.00}
\definecolor {olivedrab1}          {rgb}{0.75,1.00,0.24}
\definecolor {olivedrab2}          {rgb}{0.70,0.93,0.23}
\definecolor {olivedrab4}          {rgb}{0.41,0.55,0.13}
\definecolor {darkolivegreen1}     {rgb}{0.79,1.00,0.44}
\definecolor {darkolivegreen2}     {rgb}{0.74,0.93,0.41}
\definecolor {darkolivegreen3}     {rgb}{0.64,0.80,0.35}
\definecolor {darkolivegreen4}     {rgb}{0.43,0.55,0.24}
\definecolor {khaki1}              {rgb}{1.00,0.96,0.56}
\definecolor {khaki2}              {rgb}{0.93,0.90,0.52}
\definecolor {khaki3}              {rgb}{0.80,0.78,0.45}
\definecolor {khaki4}              {rgb}{0.55,0.53,0.31}
\definecolor {lightgoldenrod1}     {rgb}{1.00,0.93,0.55}
\definecolor {lightgoldenrod2}     {rgb}{0.93,0.86,0.51}
\definecolor {lightgoldenrod3}     {rgb}{0.80,0.75,0.44}
\definecolor {lightgoldenrod4}     {rgb}{0.55,0.51,0.30}
\definecolor {lightyellow2}        {rgb}{0.93,0.93,0.82}
\definecolor {lightyellow3}        {rgb}{0.80,0.80,0.71}
\definecolor {lightyellow4}        {rgb}{0.55,0.55,0.48}
\definecolor {yellow2}             {rgb}{0.93,0.93,0.00}
\definecolor {yellow3}             {rgb}{0.80,0.80,0.00}
\definecolor {yellow4}             {rgb}{0.55,0.55,0.00}
\definecolor {gold2}               {rgb}{0.93,0.79,0.00}
\definecolor {gold3}               {rgb}{0.80,0.68,0.00}
\definecolor {gold4}               {rgb}{0.55,0.46,0.00}
\definecolor {goldenrod1}          {rgb}{1.00,0.76,0.15}
\definecolor {goldenrod2}          {rgb}{0.93,0.71,0.13}
\definecolor {goldenrod3}          {rgb}{0.80,0.61,0.11}
\definecolor {goldenrod4}          {rgb}{0.55,0.41,0.08}
\definecolor {darkgoldenrod1}      {rgb}{1.00,0.73,0.06}
\definecolor {darkgoldenrod2}      {rgb}{0.93,0.68,0.05}
\definecolor {darkgoldenrod3}      {rgb}{0.80,0.58,0.05}
\definecolor {darkgoldenrod4}      {rgb}{0.55,0.40,0.03}
\definecolor {rosybrown1}          {rgb}{1.00,0.76,0.76}
\definecolor {rosybrown2}          {rgb}{0.93,0.71,0.71}
\definecolor {rosybrown3}          {rgb}{0.80,0.61,0.61}
\definecolor {rosybrown4}          {rgb}{0.55,0.41,0.41}
\definecolor {indianred1}          {rgb}{1.00,0.42,0.42}
\definecolor {indianred2}          {rgb}{0.93,0.39,0.39}
\definecolor {indianred3}          {rgb}{0.80,0.33,0.33}
\definecolor {indianred4}          {rgb}{0.55,0.23,0.23}
\definecolor {sienna1}             {rgb}{1.00,0.51,0.28}
\definecolor {sienna2}             {rgb}{0.93,0.47,0.26}
\definecolor {sienna3}             {rgb}{0.80,0.41,0.22}
\definecolor {sienna4}             {rgb}{0.55,0.28,0.15}
\definecolor {burlywood1}          {rgb}{1.00,0.83,0.61}
\definecolor {burlywood2}          {rgb}{0.93,0.77,0.57}
\definecolor {burlywood3}          {rgb}{0.80,0.67,0.49}
\definecolor {burlywood4}          {rgb}{0.55,0.45,0.33}
\definecolor {wheat1}              {rgb}{1.00,0.91,0.73}
\definecolor {wheat2}              {rgb}{0.93,0.85,0.68}
\definecolor {wheat3}              {rgb}{0.80,0.73,0.59}
\definecolor {wheat4}              {rgb}{0.55,0.49,0.40}
\definecolor {tan1}                {rgb}{1.00,0.65,0.31}
\definecolor {tan2}                {rgb}{0.93,0.60,0.29}
\definecolor {tan4}                {rgb}{0.55,0.35,0.17}
\definecolor {chocolate1}          {rgb}{1.00,0.50,0.14}
\definecolor {chocolate2}          {rgb}{0.93,0.46,0.13}
\definecolor {chocolate3}          {rgb}{0.80,0.40,0.11}
\definecolor {firebrick1}          {rgb}{1.00,0.19,0.19}
\definecolor {firebrick2}          {rgb}{0.93,0.17,0.17}
\definecolor {firebrick3}          {rgb}{0.80,0.15,0.15}
\definecolor {firebrick4}          {rgb}{0.55,0.10,0.10}
\definecolor {brown1}              {rgb}{1.00,0.25,0.25}
\definecolor {brown2}              {rgb}{0.93,0.23,0.23}
\definecolor {brown3}              {rgb}{0.80,0.20,0.20}
\definecolor {brown4}              {rgb}{0.55,0.14,0.14}
\definecolor {salmon1}             {rgb}{1.00,0.55,0.41}
\definecolor {salmon2}             {rgb}{0.93,0.51,0.38}
\definecolor {salmon3}             {rgb}{0.80,0.44,0.33}
\definecolor {salmon4}             {rgb}{0.55,0.30,0.22}
\definecolor {lightsalmon2}        {rgb}{0.93,0.58,0.45}
\definecolor {lightsalmon3}        {rgb}{0.80,0.51,0.38}
\definecolor {lightsalmon4}        {rgb}{0.55,0.34,0.26}
\definecolor {orange2}             {rgb}{0.93,0.60,0.00}
\definecolor {orange3}             {rgb}{0.80,0.52,0.00}
\definecolor {orange4}             {rgb}{0.55,0.35,0.00}
\definecolor {darkorange1}         {rgb}{1.00,0.50,0.00}
\definecolor {darkorange2}         {rgb}{0.93,0.46,0.00}
\definecolor {darkorange3}         {rgb}{0.80,0.40,0.00}
\definecolor {darkorange4}         {rgb}{0.55,0.27,0.00}
\definecolor {coral1}              {rgb}{1.00,0.45,0.34}
\definecolor {coral2}              {rgb}{0.93,0.42,0.31}
\definecolor {coral3}              {rgb}{0.80,0.36,0.27}
\definecolor {coral4}              {rgb}{0.55,0.24,0.18}
\definecolor {tomato2}             {rgb}{0.93,0.36,0.26}
\definecolor {tomato3}             {rgb}{0.80,0.31,0.22}
\definecolor {tomato4}             {rgb}{0.55,0.21,0.15}
\definecolor {orangered2}          {rgb}{0.93,0.25,0.00}
\definecolor {orangered3}          {rgb}{0.80,0.22,0.00}
\definecolor {orangered4}          {rgb}{0.55,0.15,0.00}
\definecolor {red2}                {rgb}{0.93,0.00,0.00}
\definecolor {red3}                {rgb}{0.80,0.00,0.00}
\definecolor {red4}                {rgb}{0.55,0.00,0.00}
\definecolor {deeppink2}           {rgb}{0.93,0.07,0.54}
\definecolor {deeppink3}           {rgb}{0.80,0.06,0.46}
\definecolor {deeppink4}           {rgb}{0.55,0.04,0.31}
\definecolor {hotpink1}            {rgb}{1.00,0.43,0.71}
\definecolor {hotpink2}            {rgb}{0.93,0.42,0.65}
\definecolor {hotpink3}            {rgb}{0.80,0.38,0.56}
\definecolor {hotpink4}            {rgb}{0.55,0.23,0.38}
\definecolor {pink1}               {rgb}{1.00,0.71,0.77}
\definecolor {pink2}               {rgb}{0.93,0.66,0.72}
\definecolor {pink3}               {rgb}{0.80,0.57,0.62}
\definecolor {pink4}               {rgb}{0.55,0.39,0.42}
\definecolor {lightpink1}          {rgb}{1.00,0.68,0.73}
\definecolor {lightpink2}          {rgb}{0.93,0.64,0.68}
\definecolor {lightpink3}          {rgb}{0.80,0.55,0.58}
\definecolor {lightpink4}          {rgb}{0.55,0.37,0.40}
\definecolor {palevioletred1}      {rgb}{1.00,0.51,0.67}
\definecolor {palevioletred2}      {rgb}{0.93,0.47,0.62}
\definecolor {palevioletred3}      {rgb}{0.80,0.41,0.54}
\definecolor {palevioletred4}      {rgb}{0.55,0.28,0.36}
\definecolor {maroon1}             {rgb}{1.00,0.20,0.70}
\definecolor {maroon2}             {rgb}{0.93,0.19,0.65}
\definecolor {maroon3}             {rgb}{0.80,0.16,0.56}
\definecolor {maroon4}             {rgb}{0.55,0.11,0.38}
\definecolor {violetred1}          {rgb}{1.00,0.24,0.59}
\definecolor {violetred2}          {rgb}{0.93,0.23,0.55}
\definecolor {violetred3}          {rgb}{0.80,0.20,0.47}
\definecolor {violetred4}          {rgb}{0.55,0.13,0.32}
\definecolor {magenta2}            {rgb}{0.93,0.00,0.93}
\definecolor {magenta3}            {rgb}{0.80,0.00,0.80}
\definecolor {magenta4}            {rgb}{0.55,0.00,0.55}
\definecolor {orchid1}             {rgb}{1.00,0.51,0.98}
\definecolor {orchid2}             {rgb}{0.93,0.48,0.91}
\definecolor {orchid3}             {rgb}{0.80,0.41,0.79}
\definecolor {orchid4}             {rgb}{0.55,0.28,0.54}
\definecolor {plum1}               {rgb}{1.00,0.73,1.00}
\definecolor {plum2}               {rgb}{0.93,0.68,0.93}
\definecolor {plum3}               {rgb}{0.80,0.59,0.80}
\definecolor {plum4}               {rgb}{0.55,0.40,0.55}
\definecolor {mediumorchid1}       {rgb}{0.88,0.40,1.00}
\definecolor {mediumorchid2}       {rgb}{0.82,0.37,0.93}
\definecolor {mediumorchid3}       {rgb}{0.71,0.32,0.80}
\definecolor {mediumorchid4}       {rgb}{0.48,0.22,0.55}
\definecolor {darkorchid1}         {rgb}{0.75,0.24,1.00}
\definecolor {darkorchid2}         {rgb}{0.70,0.23,0.93}
\definecolor {darkorchid3}         {rgb}{0.60,0.20,0.80}
\definecolor {darkorchid4}         {rgb}{0.41,0.13,0.55}
\definecolor {purple1}             {rgb}{0.61,0.19,1.00}
\definecolor {purple2}             {rgb}{0.57,0.17,0.93}
\definecolor {purple3}             {rgb}{0.49,0.15,0.80}
\definecolor {purple4}             {rgb}{0.33,0.10,0.55}
\definecolor {mediumpurple1}       {rgb}{0.67,0.51,1.00}
\definecolor {mediumpurple2}       {rgb}{0.62,0.47,0.93}
\definecolor {mediumpurple3}       {rgb}{0.54,0.41,0.80}
\definecolor {mediumpurple4}       {rgb}{0.36,0.28,0.55}
\definecolor {thistle1}            {rgb}{1.00,0.88,1.00}
\definecolor {thistle2}            {rgb}{0.93,0.82,0.93}
\definecolor {thistle3}            {rgb}{0.80,0.71,0.80}
\definecolor {thistle4}            {rgb}{0.55,0.48,0.55}
\definecolor {gray1}               {rgb}{0.01,0.01,0.01}
\definecolor {gray2}               {rgb}{0.02,0.02,0.02}
\definecolor {gray3}               {rgb}{0.03,0.03,0.03}
\definecolor {gray4}               {rgb}{0.04,0.04,0.04}
\definecolor {gray5}               {rgb}{0.05,0.05,0.05}
\definecolor {gray6}               {rgb}{0.06,0.06,0.06}
\definecolor {gray7}               {rgb}{0.07,0.07,0.07}
\definecolor {gray8}               {rgb}{0.08,0.08,0.08}
\definecolor {gray9}               {rgb}{0.09,0.09,0.09}
\definecolor {gray10}              {rgb}{0.10,0.10,0.10}
\definecolor {gray11}              {rgb}{0.11,0.11,0.11}
\definecolor {gray12}              {rgb}{0.12,0.12,0.12}
\definecolor {gray13}              {rgb}{0.13,0.13,0.13}
\definecolor {gray14}              {rgb}{0.14,0.14,0.14}
\definecolor {gray15}              {rgb}{0.15,0.15,0.15}
\definecolor {gray16}              {rgb}{0.16,0.16,0.16}
\definecolor {gray17}              {rgb}{0.17,0.17,0.17}
\definecolor {gray18}              {rgb}{0.18,0.18,0.18}
\definecolor {gray19}              {rgb}{0.19,0.19,0.19}
\definecolor {gray20}              {rgb}{0.20,0.20,0.20}
\definecolor {gray21}              {rgb}{0.21,0.21,0.21}
\definecolor {gray22}              {rgb}{0.22,0.22,0.22}
\definecolor {gray23}              {rgb}{0.23,0.23,0.23}
\definecolor {gray24}              {rgb}{0.24,0.24,0.24}
\definecolor {gray25}              {rgb}{0.25,0.25,0.25}
\definecolor {gray26}              {rgb}{0.26,0.26,0.26}
\definecolor {gray27}              {rgb}{0.27,0.27,0.27}
\definecolor {gray28}              {rgb}{0.28,0.28,0.28}
\definecolor {gray29}              {rgb}{0.29,0.29,0.29}
\definecolor {gray30}              {rgb}{0.30,0.30,0.30}
\definecolor {gray31}              {rgb}{0.31,0.31,0.31}
\definecolor {gray32}              {rgb}{0.32,0.32,0.32}
\definecolor {gray33}              {rgb}{0.33,0.33,0.33}
\definecolor {gray34}              {rgb}{0.34,0.34,0.34}
\definecolor {gray35}              {rgb}{0.35,0.35,0.35}
\definecolor {gray36}              {rgb}{0.36,0.36,0.36}
\definecolor {gray37}              {rgb}{0.37,0.37,0.37}
\definecolor {gray38}              {rgb}{0.38,0.38,0.38}
\definecolor {gray39}              {rgb}{0.39,0.39,0.39}
\definecolor {gray40}              {rgb}{0.40,0.40,0.40}
\definecolor {gray42}              {rgb}{0.42,0.42,0.42}
\definecolor {gray43}              {rgb}{0.43,0.43,0.43}
\definecolor {gray44}              {rgb}{0.44,0.44,0.44}
\definecolor {gray45}              {rgb}{0.45,0.45,0.45}
\definecolor {gray46}              {rgb}{0.46,0.46,0.46}
\definecolor {gray47}              {rgb}{0.47,0.47,0.47}
\definecolor {gray48}              {rgb}{0.48,0.48,0.48}
\definecolor {gray49}              {rgb}{0.49,0.49,0.49}
\definecolor {gray50}              {rgb}{0.50,0.50,0.50}
\definecolor {gray51}              {rgb}{0.51,0.51,0.51}
\definecolor {gray52}              {rgb}{0.52,0.52,0.52}
\definecolor {gray53}              {rgb}{0.53,0.53,0.53}
\definecolor {gray54}              {rgb}{0.54,0.54,0.54}
\definecolor {gray55}              {rgb}{0.55,0.55,0.55}
\definecolor {gray56}              {rgb}{0.56,0.56,0.56}
\definecolor {gray57}              {rgb}{0.57,0.57,0.57}
\definecolor {gray58}              {rgb}{0.58,0.58,0.58}
\definecolor {gray59}              {rgb}{0.59,0.59,0.59}
\definecolor {gray60}              {rgb}{0.60,0.60,0.60}
\definecolor {gray61}              {rgb}{0.61,0.61,0.61}
\definecolor {gray62}              {rgb}{0.62,0.62,0.62}
\definecolor {gray63}              {rgb}{0.63,0.63,0.63}
\definecolor {gray64}              {rgb}{0.64,0.64,0.64}
\definecolor {gray65}              {rgb}{0.65,0.65,0.65}
\definecolor {gray66}              {rgb}{0.66,0.66,0.66}
\definecolor {gray67}              {rgb}{0.67,0.67,0.67}
\definecolor {gray68}              {rgb}{0.68,0.68,0.68}
\definecolor {gray69}              {rgb}{0.69,0.69,0.69}
\definecolor {gray70}              {rgb}{0.70,0.70,0.70}
\definecolor {gray71}              {rgb}{0.71,0.71,0.71}
\definecolor {gray72}              {rgb}{0.72,0.72,0.72}
\definecolor {gray73}              {rgb}{0.73,0.73,0.73}
\definecolor {gray74}              {rgb}{0.74,0.74,0.74}
\definecolor {gray75}              {rgb}{0.75,0.75,0.75}
\definecolor {gray76}              {rgb}{0.76,0.76,0.76}
\definecolor {gray77}              {rgb}{0.77,0.77,0.77}
\definecolor {gray78}              {rgb}{0.78,0.78,0.78}
\definecolor {gray79}              {rgb}{0.79,0.79,0.79}
\definecolor {gray80}              {rgb}{0.80,0.80,0.80}
\definecolor {gray81}              {rgb}{0.81,0.81,0.81}
\definecolor {gray82}              {rgb}{0.82,0.82,0.82}
\definecolor {gray83}              {rgb}{0.83,0.83,0.83}
\definecolor {gray84}              {rgb}{0.84,0.84,0.84}
\definecolor {gray85}              {rgb}{0.85,0.85,0.85}
\definecolor {gray86}              {rgb}{0.86,0.86,0.86}
\definecolor {gray87}              {rgb}{0.87,0.87,0.87}
\definecolor {gray88}              {rgb}{0.88,0.88,0.88}
\definecolor {gray89}              {rgb}{0.89,0.89,0.89}
\definecolor {gray90}              {rgb}{0.90,0.90,0.90}
\definecolor {gray91}              {rgb}{0.91,0.91,0.91}
\definecolor {gray92}              {rgb}{0.92,0.92,0.92}
\definecolor {gray93}              {rgb}{0.93,0.93,0.93}
\definecolor {gray94}              {rgb}{0.94,0.94,0.94}
\definecolor {gray95}              {rgb}{0.95,0.95,0.95}
\definecolor {gray97}              {rgb}{0.97,0.97,0.97}
\definecolor {gray98}              {rgb}{0.98,0.98,0.98}
\definecolor {gray99}              {rgb}{0.99,0.99,0.99}
\definecolor {darkgrey}            {rgb}{0.66,0.66,0.66}

\newcommand{\RSNOTE}[1]{\marginpar{\textcolor{darkgreen}{\textbf{RS: }
      {\footnotesize #1}}}}
\renewcommand{\RSNOTE}[1]{}

\newcommand{\smalltt}[1]{{\tt\footnotesize #1}}

\begin{document}

\title{
\LARGE
Software Model Checking via Large-Block Encoding
}

\author{
  Dirk Beyer \quad
  Alessandro Cimatti \quad
  Alberto Griggio \\
  M. Erkan Keremoglu \quad
  Roberto Sebastiani 
}

\reportnumber{SFU-CS-2009-09}
\date{April 29, 2009}

\makecover

\phantom{x}

\title{
\LARGE
Software Model Checking via Large-Block Encoding\,$^{\text{{\scriptsize 1}}}$
}

\author{
\begin{tabular}{ccccc}
           Dirk Beyer\,$^{\text{{\scriptsize 2}}}$ 
 & \!\!\!\!Alessandro Cimatti\,$^{\text{{\scriptsize 3}}}$ 
 & \!\!\!\!\!\!Alberto Griggio\,$^{\text{{\scriptsize 2,4}}}$ 
 & \!\!\!\!\!\!M. Erkan Keremoglu\,$^{\text{{\scriptsize 2}}}$ 
 & \!\!\!\!Roberto Sebastiani\,$^{\text{{\scriptsize 4}}}$ \\
           \small Simon Fraser Univ.
 & \!\!\!\!\small FBK-irst
 & \!\!\!\!\!\!\small Univ. of Trento \& Simon Fraser Univ.
 & \!\!\!\!\!\!\small Simon Fraser Univ.
 & \!\!\!\!\small Univ. of Trento
\end{tabular}
}

\markboth{}{}

\maketitle

\setcounter{page}{1}
\pagestyle{headings}

\footnotetext[1]{
 Technical Report SFU-CS-2009-09, DISI-09-026, FBK-irst-2009.04.005.
}

\footnotetext[2]{
 Supported in part by the Canadian NSERC grant RGPIN 341819-07 and by the SFU grant PRG 06-3.
$^{\text{{\scriptsize 3}}}$%
 Supported in part by the European Commission
 grant FP7-2007-IST-1-217069 COCONUT.
$^{\text{{\scriptsize 4}}}$%
 Supported in part by the SRC/GRC grant 2009-TJ-1880 WOLFLING and by the MIUR grant PRIN 20079E5KM8\_002.
}

\begin{abstract}
The construction and analysis of an abstract reachability tree (ART)
are the basis for a successful method for software verification.
The ART represents unwindings of the control-flow graph of the program.
Traditionally, a transition of the ART represents a single block of the program,
and therefore, we call this approach single-block encoding (SBE).
SBE may result in a huge number of program paths to be explored,
which constitutes a fundamental source of inefficiency.
We propose a generalization of the approach, 
in which transitions of the ART represent larger portions of the program;
we call this approach large-block encoding (LBE).
LBE may reduce the number of paths to be explored up to exponentially.
Within this framework, we also investigate symbolic representations: 
for representing abstract states, 
in addition to conjunctions as used in SBE, 
we investigate the use of arbitrary Boolean formulas; 
for computing abstract-successor states, 
in addition to Cartesian predicate abstraction as used in SBE, 
we investigate the use of Boolean predicate abstraction.
The new encoding leverages the efficiency of state-of-the-art SMT
solvers, which can symbolically compute abstract large-block successors.
Our experiments on benchmark C programs show that the
large-block encoding outperforms the single-block encoding.

\end{abstract}

\section{Introduction}
\label{sec-intro}

Software model checking is an effective technique for software verification. 
Several advances in the field have lead to tools that are able to verify
programs of considerable size, and show significant advantages over traditional
techniques in terms of precision of the analysis
(e.g., \slam~\cite{SLAM} and \blast~\cite{BLAST}).
However, efficiency and scalability remain major concerns in software model checking
and hamper the adaptation of the techniques in industrial practice.
A successful approach to software model checking
is based on the construction and analysis of an abstract reachability tree (ART),
and predicate abstraction is one of the favorite abstract domains.
The ART represents unwindings of the control-flow graph of the program.
The search is usually guided by the control flow of the program.
Nodes of the ART typically consist of the control-flow location,
the call stack, and formulas that represent the data states.
During the refinement process, the ART nodes are incrementally refined.

In the traditional ART approach, each program operation 
(assignment operation, assume operation, function call, function return)
is represented by a single edge in the ART.
Therefore, we call this approach \emph{single-block encoding} (SBE).
A fundamental source of inefficiency of this approach is the fact that the control-flow
of the program can induce a huge number of paths (and nodes) in the ART,
which are explored independently of each other.
 
We propose a novel, broader view on ART-based software model checking, 
where a much more compact abstract space is used,
resulting thus in a much smaller number of paths to be enumerated in the ART.
Instead of using edges that represent single program operations,
we encode entire parts of the program in one edge.
In contrast to SBE, we call our new approach \emph{large-block encoding} (LBE).
In general, the new encoding may result in an exponential reduction of the number of
ART nodes.

The generalization from SBE to LBE has two main consequences.
First, LBE requires a more general representation of
abstract states than SBE.
SBE is typically based on mere \emph{conjunctions} of predicates. 
Because the LBE approach summarizes large portions of the control flow,
conjunctions are not sufficient, and we need to use
\emph{arbitrary Boolean combinations} of predicates
to represent the abstract states.
Second, LBE requires a more accurate abstraction in
the abstract-successor computations.
Intuitively, an abstract edge represents many different
paths of the program,
and therefore it is necessary that the abstract-successor computations take the
relationships between the predicates into account.

In order to make this generalization practical, we rely on
efficient solvers for satisfiability modulo theories (SMT). 
In particular, enabling factors are the capability of performing
Boolean reasoning efficiently (e.g.,~\cite{seba_lazy_smt}), 
the availability of effective algorithms for abstraction computation
(e.g.,~\cite{allsmt,bddsmt}), 
and interpolation procedures to extract new predicates~\cite{TACAS08,CSIsat}. 

Considering Boolean abstraction and large-block encoding in addition to the
traditional techniques, we obtain the following interesting observations:
(i) whilst the SBE approach requires a large number of successor computations,
the LBE approach reduces the number of successor computations dramatically (possibly exponentially);
(ii) whilst Cartesian abstraction can be efficiently
computed with a linear number of SMT solver queries, 
Boolean abstraction is expensive to compute because it requires an enumeration
of all satisfiable assignments for the predicates.
Therefore, two combinations of the above strategies provide an interesting tradeoff:
The combination of SBE with Cartesian abstraction was successfully 
implemented by tools like \blast and \slam.
We investigate the combination of LBE with Boolean abstraction,
by first formally defining LBE in terms of a summarization of the control-flow automaton for the program,
and then implementing this LBE approach together with a Boolean predicate abstraction.
We evaluate the performance and precision by comparing it 
with the model checker \blast and with an own implementation
of the traditional approach.
Our own implementation of the SBE and LBE approach is 
integrated as a new component into \cpachecker~\cite{CPACHECKER}%
\footnote{Available at {\tt\footnotesize http://www.cs.sfu.ca/$\sim$dbeyer/CPAchecker}}.
The experiments show that our new approach outperforms the previous approach.

\begin{figure*}[t]
  \centering
  \newcommand{\subcaption}[1]{\fontsize{8}{10}\selectfont #1\normalsize}
  \begin{tabular}{lll}
    \raisebox{2mm}{\scalebox{0.7}{
        \begin{minipage}[b]{6cm}
\begin{verbatim}
L1:  if(p1) {
L2:    x1 = 1;
     }
L3:  if(p2) {
L4:    x2 = 2;
     }
L5:  if(p3) {
L6:    x3 = 3;
     }
L7:  if(p1) {
L8:    if (x1 != 1) goto ERR;
     }
L9:  if (p2) {
L10:   if (x2 != 2) goto ERR;
     }
L11: if (p3) {
L12:   if (x3 != 3) goto ERR;
     }
L13: return EXIT_SUCCESS;
ERR: return EXIT_FAILURE;
\end{verbatim}
        \end{minipage}
      }}
    &
    \raisebox{0.4cm}{\includegraphics[scale=0.09]{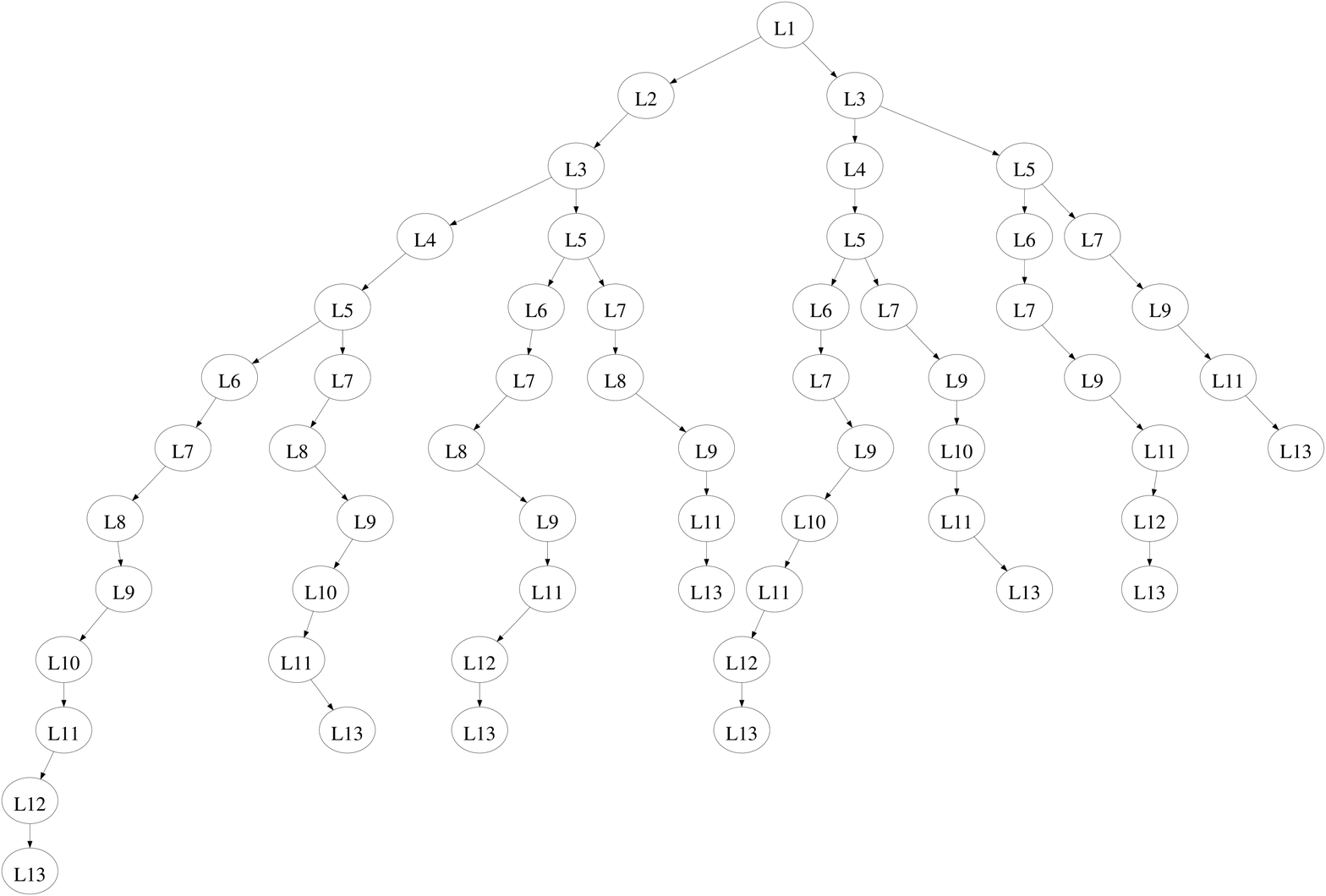}} 
    &
    \raisebox{0.2cm}{\includegraphics[scale=0.25]{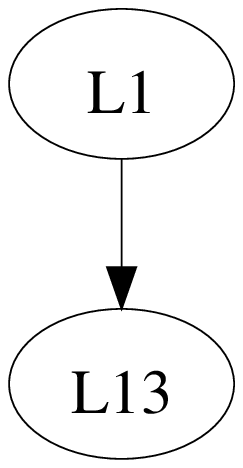}}
    \\
    \subcaption{(a) Example C program}
    &
    \subcaption{\hspace{1cm}(b) ART for SBE}
    &
    \subcaption{\hspace{-1.3cm}(c) ART for LBE}
  \end{tabular}
  \caption{
    \label{fig:example_blowup}
Example program and corresponding ARTs for SBE and LBE;
this example was mentioned as verification challenge for ART-based approaches 
by several colleagues.
}
\end{figure*}

\begin{example}
We illustrate the advantage of LBE over SBE
on the example program 
in Fig.~\ref{fig:example_blowup} (a).
In SBE, each program location is modeled explicitly,
and an abstract-successor computation is performed for each program operation.
Figure~\ref{fig:example_blowup}~(b) shows the structure of the resulting ART.
In the figure, abstract states are drawn as ellipses,
and labeled with the location of the abstract state;
the arrows indicate that there exists an edge from the source location
to the target location in the control-flow.
The ART represents all feasible program paths.
For example, the leftmost program path is taking the `then' branch of every `if' statement.
For every edge in the ART, an abstract-successor computation is performed,
which potentially includes several SMT solver queries.
The problems given to the SMT solver are usually very small,
and the runtime sums up over a large amount of simple queries.
Therefore, model checkers that are based on SBE (like \blast) 
experience serious performance problems on programs with such an exploding structure
(cf.~the {\small \tt test\_locks} examples in Table~\ref{tab:eval}).
In LBE, the control-flow graph is summarized,
such that control-flow edges represent entire subgraphs of the original control-flow.
In our example, most of the program is 
summarized into one control-flow edge.
Figure~\ref{fig:example_blowup}~(c) shows the structure of the resulting ART,
in which all the feasible paths of the program are represented by a single edge.
The exponential growth of the ART does not occur.
\end{example}

\smallsec{Related Work} 
The model checkers \slam and \blast are typical examples 
for the SBE approach~\cite{SLAM,BLAST},
both based on counterexample-guided abstraction refinement (CEGAR)~\cite{ClarkeCEGAR}.
Also the tool \satabs is based on CEGAR, 
but it performs a fully symbolic search in the abstract space~\cite{SATABS}.
In contrast, our approach still follows the lazy-abstraction
paradigm~\cite{LazyAbstraction}, 
but it abstracts and refines chunks of the program ``on-the-fly''.
The work of McMillan is also based on lazy abstraction, 
but instead of using predicate abstraction for the abstract domain,
Craig interpolants from infeasible error paths are directly used,
thus avoiding abstract-successor computations~\cite{McMillanCAV06}. 
A fundamentally different approach to software model checking
is bounded model checking (BMC), with the most prominent example CBMC \cite{CBMC}.
Programs are unrolled up to a given depth, and a formula is constructed
which is satisfiable iff one of the considered program
executions reaches a certain error location.
The analysis tool \calysto is an example of an ``extended static checker'',
following an approach similar to BMC
when generating verification conditions~\cite{CALYSTO}, 
while possibly abstracting away some irrelevant parts of the program.
The BMC approaches are targeted towards discovering bugs, 
and can not be used to prove program safety.

\smallsec{Structure}
Section~\ref{sec-background} provides the necessary background.
Section~\ref{sec-summarization} explains our contribution in detail.
We experimentally evaluate our novel
approach in Sect.~\ref{sec-experiments}. 
In Sect.~\ref{sec-conclusion}, we draw some conclusions and outline
directions for future research.

\markboth
{Beyer, Cimatti, Griggio, Keremoglu, Sebastiani: \quad
Software Model Checking via Large-Block Encoding}
{Beyer, Cimatti, Griggio, Keremoglu, Sebastiani: \quad
Software Model Checking via Large-Block Encoding}

\section{Background}
\label{sec-background}

\subsection{Programs and Control-Flow Automata}

We restrict the presentation to
a simple imperative programming language,
where all operations are either assignments or assume operations,
and all variables range over integers.%
\footnote{Our implementation is based on \cpachecker{},
which operates on C~programs that are given in
the {\sc Cil} intermediate language~\cite{CIL};
function calls are supported.}
We represent a program by a \emph{control-flow automaton} (CFA).
A CFA $A = (\locs, G)$ consists of
a set~$\locs$ of program locations, which model the program counter $\pc$
and a set $G \subseteq \locs \times Ops \times \locs$ of control-flow edges,
which model the operations that are executed when
control flows from one program location to another.
The set of variables that occur in operations from~$Ops$
is denoted by~$X$.
A~\emph{program}~$P = (A, \pci, \pct)$ consists of 
a CFA $A = (\locs, G)$ (which models the control flow of the program),
an initial program location~$\pci \in \locs$ (which models the program entry)
such that $G$ does not contain any edge $(\cdot, \cdot, \pci)$, and
a target program location~$\pct \in \locs$ (which models the error location).

A \emph{concrete data state} of a program is
a variable assignment $c: X \to \Ints$
that assigns to each variable an integer value.
The set of all concrete data states of a program is denoted by~\concr.
A set~$r \subseteq \concr$ of concrete data states is called \emph{region}.
We represent regions using first-order formulas (with free variables from $X$): 
a formula $\phi$ represents the set $S$ of all data states $c$ that imply it (i.e. $S = \{c \mid c \models \phi\}$).
A~\emph{concrete state} of a program is a pair~$(l, c)$
where $l \in \locs$ is a program location and $c$~is a concrete data state.
A~pair~$(l, \phi)$ represents the following set of all concrete states:
$\{ (l, c) \mid c \models \phi\}$.
The \emph{concrete semantics} of an operation $\op \in Ops$ is defined by the strongest postcondition operator~$\SP_\op$: for a formula $\phi$, $\SP_\op(\phi)$ represents the set of data states that are reachable from any of the states in region represented by $\phi$ after the execution of $\op$.
Given a formula~$\phi$ that represents a set of concrete data states,
for an 
assignment operation~$s := e$,
we have 
${\SP_{s := e}(\phi) =
    \exists \widehat{s}:
    \phi_{[s \mapsto \widehat{s}]} \land 
    (s = e_{[s \mapsto \widehat{s}]})}$;
and for an 
assume operation~$\mathit{assume}(p)$,
we have 
$\SP_{\mathit{assume}(p)}(\phi) = \phi \land p$.

A \emph{path}~$\path$ is a sequence~$\seq{(\op_1, \pc_1), ..., (\op_n, \pc_n)}$
of pairs of operations and locations.
The path~$\path$ is called \emph{program path} if for every $i$ with $1 \leq i \leq n$
there exists a CFA edge~$g = (\pc_{i-1}, \op_i, \pc_i)$,
i.e., $\path$ represents a syntactical walk through the CFA.
The \emph{concrete semantics for a program path}
$\path = \seq{(\op_1, \pc_1), ..., (\op_n, \pc_n)}$
is defined as the successive application of the strongest postoperator
for each operation: 
$\SP_\path(\phi) = \SP_{\op_n}( ... \SP_{\op_1}(\phi) ... )$. 
The set of concrete states that result from running $\path$
is represented by the pair~$(l_n, \SP_\path(\true))$.
A program path $\path$ is \emph{feasible} if $\SP_\path(\true)$ is satisfiable.
 A concrete state~$(l_n, c_n)$ is called \emph{reachable} 
if there exists a feasible program path $\path$
whose final location is $l_n$ and such that $c_n \models \SP_\path(\true)$.
A location $l$ is reachable if there exists a concrete state $c$ such that $(l, c)$ is reachable.
A program is \emph{safe} if $\pct$ is not reachable.

\subsection{Predicate Abstraction}

Let $\preds$ be a set of predicates over program variables in
a quantifier-free theory~${\cal T}$.
A~\emph{formula}~$\phi$ is a Boolean combination of predicates from~$\preds$.
A~\emph{precision for a formula} is a finite subset~$\prec \subset \preds$ of predicates.

\smallsec{Cartesian Predicate Abstraction}
Let~$\prec$ be a precision.
The \emph{Cartesian predicate abstraction~$\abs{\phi}_\cart^\prec$ of a formula~$\phi$}
is the strongest conjunction of predicates from~$\prec$ 
entailed by $\phi$:
$\abs{\phi}_\cart^\prec := \bigwedge~ \{p \in \prec \mid \phi \implies p\}$.
Such a predicate abstraction of a formula~$\phi$
that represents a region of concrete program states,
is used as an \emph{abstract state} (i.e., an abstract representation of the region)
in program verification.
For a formula $\phi$ and a precision $\prec$,
the Cartesian predicate abstraction~$\abs{\phi}_\cart^\prec$ of~$\phi$
can be computed by $|\prec|$~SMT-solver queries.
The abstract strongest postoperator~$\SP^\prec$ 
for a predicate abstraction~$\prec$
transforms the abstract state $\abs{\phi}^\prec_\cart$ into 
its successor~$\abs{\phi'}^\prec_\cart$ for a program operation~$\op$,
written as $\abs{\phi'}^\prec_\cart = \SP^\prec_\op(\abs{\phi}^\prec_\cart)$, if
$\abs{\phi'}^\prec_\cart$
is the Cartesian predicate abstraction of $\SP_\op(\abs{\phi}^\prec_\cart)$,
i.e., $\abs{\phi'}^\prec_\cart = (\SP_\op(\abs{\phi}^\prec_\cart))_\cart^\prec$.
For more details, we refer the reader to the work of Ball et al.~\cite{BPR01}.

\smallsec{Boolean Predicate Abstraction}
Let~$\prec$ be a precision.
The \emph{Boolean predicate abstraction~$\abs{\phi}_\bool^\prec$ of a formula~$\phi$}
is the strongest Boolean combination of predicates from~$\prec$ 
that is entailed by $\phi$.
For a formula~$\phi$ and a precision~$\prec$,
the Boolean predicate abstraction~$\abs{\phi}_\bool^\prec$ of~$\phi$ 
can be computed by querying an SMT solver in the following way:
For each predicate~$p_i \in \prec$, we introduce a propositional variable~$v_i$.
Now we ask an SMT solver to enumerate all satisfying assignments
of~$v_1, ..., v_{|\prec|}$
in the formula
$\phi \land \bigwedge_{p_i \in \prec} (p_i \iff v_i)$.
For each satisfying assignment, we construct a conjunction
of all predicates from $\prec$ whose corresponding propositional
variable occurs positive in the assignment.
The disjunction of all such conjunctions is the Boolean predicate abstraction for~$\phi$.
The abstract strongest postoperator~$\SP^\prec$ 
for a predicate abstraction~$\prec$
transforms the abstract state $\abs{\phi}^\prec_\bool$ into 
its successor~$\abs{\phi'}^\prec_\bool$ for a program operation~$\op$,
written as $\abs{\phi'}^\prec_\bool = \SP^\prec_\op(\abs{\phi}^\prec_\bool)$, if
$\abs{\phi'}^\prec_\bool$
is the Boolean predicate abstraction of $\SP_\op(\abs{\phi}^\prec_\bool)$,
i.e., $\abs{\phi'}^\prec_\bool = (\SP_\op(\abs{\phi}^\prec_\bool))_\bool^\prec$.
For more details, we refer the reader to the work of Lahiri et al.~\cite{allsmt}.

\subsection{ART-based Software Model Checking with SBE}

The \emph{precision for a program} is a function~$\PREC: L \to 2^\preds$,
 which assigns to each program location a precision for a formula.
An ART-based algorithm for software model checking
takes an initial precision $\PREC$ (which is typically very coarse) for the predicate abstraction,
and constructs an ART for the input program and $\PREC$.
An ART is a tree whose nodes are labeled with program locations and
abstract states~\cite{BLAST} (i.e., $n=(l,\phi)$).
For a given ART node, all children nodes are labeled with successor locations and
abstract successor states, according to the strongest postoperator and the predicate abstraction.
A node~$n=(l,\phi)$ is called \emph{covered} if there exists another
ART node $n'=(l,\phi')$ that entails~$n$ (i.e., s.t. $\phi'\models\phi$).
An ART is called \emph{complete} if every node is either covered or all possible abstract successor states
are present in the ART as children of the node.
If a complete ART is constructed and the ART does not contain any error node,
then the program is considered correct \cite{LazyAbstraction}.
If the algorithm adds an error node to the ART, then
the corresponding path $\path$ is checked to determine if $\path$ is feasible
(i.e., if the corresponding concrete program path is executable)
or infeasible (i.e., if there is no corresponding program execution).
In the former case the path represents a witness for a program bug.
In the latter case the path is analyzed, and a refinement $\PREC'$ of $\PREC$ is generated,
such that the same path cannot occur again during the ART exploration.
The concept of using an infeasible error path for abstraction refinement
is called counterexample-guided abstraction refinement (CEGAR)~\cite{ClarkeCEGAR}.
The concept of iteratively constructing an ART and refining
only the precisions along the considered path is called lazy abstraction~\cite{LazyAbstraction}.
Craig interpolation is a successful approach to 
predicate extraction for refinement~\cite{AbstractionsFromProofs}.
After the refining the precision, the algorithm
continues with the next iteration, using $\PREC'$ instead of $\PREC$ to construct the ART,
until either a complete error-free ART is obtained, or an error is found
(note that the procedure might not terminate).
For more details and a more in-depth illustration of the overall ART algorithm, 
we refer the reader to the \blast article~\cite{BLAST}.

In order to make the algorithm scale on practical examples,
implementations such as \blast or \slam
use the simple but coarse Cartesian abstraction,
instead of the expensive but precise Boolean abstraction.
Despite its potential imprecision, Cartesian abstraction has been proved successful
for the verification of many real-world programs.
In the SBE approach, given the large number of successor computations,
the computation of the Boolean predicate abstraction is in fact
too expensive, as it may require an SMT solver to enumerate
an exponential number of assignments on the predicates in the precision,
for each single successor computation.
The reason for the success of Cartesian abstraction if used together with SBE, 
is that for a given program path,
state overapproximations that are expressible as conjunctions of atomic predicates 
---for which Boolean and Cartesian abstractions are equivalent---
are often good enough to prove that the error location is not reachable in the abstract space.

\section{Large-Block Encoding}
\label{sec-summarization}

\subsection{Summarization of Control-Flow Automata}

The first, main step of LBE is 
the summarization of the program CFA, 
in which each large control-flow subgraph 
that is free of loops
is replaced by a single control-flow edge with a large formula that
represents the removed subgraph. 
This process, which we call {\em summarization} of the CPA, 
consists of the fixpoint application of three rewriting rules that we describe
below: first we apply Rule~0 once, and then we    
repeatedly apply Rules 1 and 2, until no rule is applicable anymore.

Let~$P = (A, \pci, \pct)$ be a program with~CFA $A = (\locs, G)$.

\smallsec{Rule 0 (Error Sink)}
We remove all edges~$(\pct, \cdot, \cdot)$ from G,
i.e., the target location~$\pct$ becomes a sink node with no outgoing edges.

\smallsec{Rule 1 (Sequence)}
\newcommand{\Gout}[1]{\ensuremath{G_{#1}^\rightarrow}}
If $G$ contains an edge $(l_1, \op_1, l_2)$ with $l_1 \neq l_2$ 
\begin{wrapfigure}[7]{r}{4cm}
  \vspace{-5mm}
  \scalebox{0.4}{%
\ifpdf
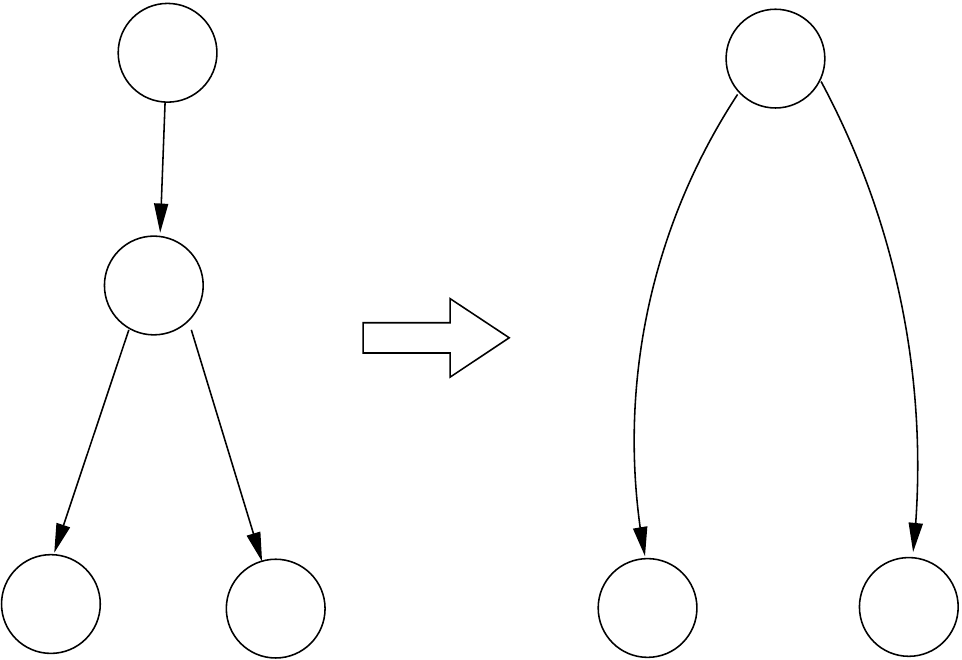
\else
\input{rule1_example.pstex_t}
\fi}
\end{wrapfigure}
and no other incoming edges for $l_2$ (i.e. edges $(\cdot, \cdot, l_2)$),
and $\Gout{l_2}$ is the subset of $G$ of outgoing edges for $l_2$,
then we change the CFA~$A$ in the following way:
(1) we remove location~$l_2$ from~$\locs$, and
(2) we remove the edges $(l_1, \op_1, l_2)$ and all the edges in $\Gout{l_2}$ from $G$,
and for each edge $(l_2, \op_i, l_i) \in \Gout{l_2}$, we add the edge 
$(l_1, \op_1 \sequence \op_i, l_i)$ to $G$,
where 
$\SP_{\op_1 \sequence \op_i}(\phi) = \SP_{\op_i}(\SP_{\op_1}(\phi))$.
(Note that $\Gout{l_2}$ might contain an edge $(l_2, \cdot, l_1)$.)

\smallsec{Rule 2 (Choice)}
If $\locs_2 = \{l_1, l_2\}$ and $A_{|\locs_2} = (\locs_2, G_2)$ 
\begin{wrapfigure}[5]{r}{3.5cm}
  \vspace{-3mm}
  \scalebox{0.4}{%
\ifpdf
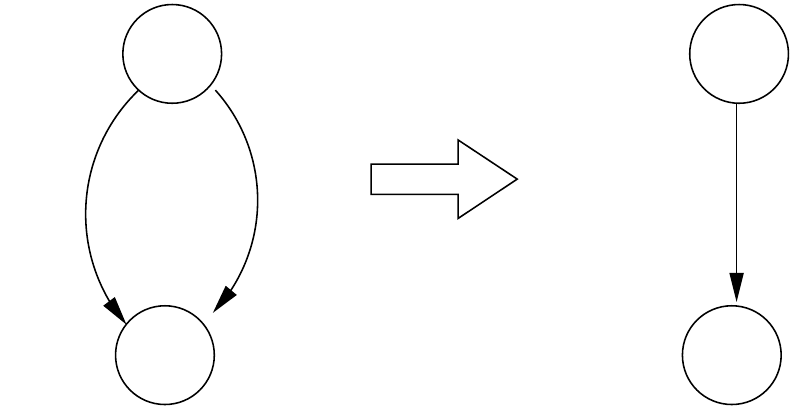
\else
\input{rule2_example.pstex_t}
\fi}
\end{wrapfigure}
is the subgraph of~$A$ with nodes from~$\locs_2$
and the set~$G_2$ of edges contains the two edges
$(l_1, \op_1, l_2)$ and $(l_1, \op_2, l_2)$,
then we change the CFA~$A$ in the following way:
(1) we remove the two edges $(l_1, \op_1, l_2)$ and $(l_1, \op_2, l_2)$ from G
    and add the edge $(l_1, \op_1 \parallel \op_2, l_2)$ to G,
where 
$\SP_{\op_1 \parallel \op_2}(\phi) = \SP_{\op_1}(\phi) \lor \SP_{\op_2}(\phi)$.
(Note that there might be a backwards edge $(l_2, \cdot, l_1)$.)

Let $P = (A, \pci, \pct)$ be a program
and let $A'$ be another CFA for~$P$.
The CFA~$A'$ is the \emph{summarization} of $A$ if
$A'$ is obtained from $A$ via stepwise application of the two rules,
and none of the two rules can be further applied.

\begin{example}
Figure~\ref{fig:rules_example2} shows a program (a) and its corresponding CFA (b).
The control-flow automaton (CFA) is stepwise transformed to its summarization CFA (h)
as follows:
Rule~1 eliminates location 6 to (c),
Rule~1 eliminates location 3 to (d),
Rule~1 eliminates location 4 to (e),
Rule~2 eliminates one edge 2--5 to (f),
Rule~1 eliminates location 5 to (g),
Rule~1 eliminates location 2 to (h).
\end{example}

\begin{figure*}[t]
\newcommand{\exampleprog}{
\raisebox{-5cm}{
  \parbox{20em}{
{\tt L1:~while~(i>0)~\{}\\
{\tt L2:~~~if~(x==1)~\{}\\
{\tt L3:~~~~~z~=~0;}\\
{\tt \phantom{L0:}~~~\}~else~\{}\\
{\tt L4:~~~~~z~=~1;}\\
{\tt \phantom{L0:}~~~\}}\\
{\tt L5:~~~i = i-1;}\\
{\tt L6:~\}}\\
}
}
}
\centering
   \scalebox{0.3}{%
\ifpdf
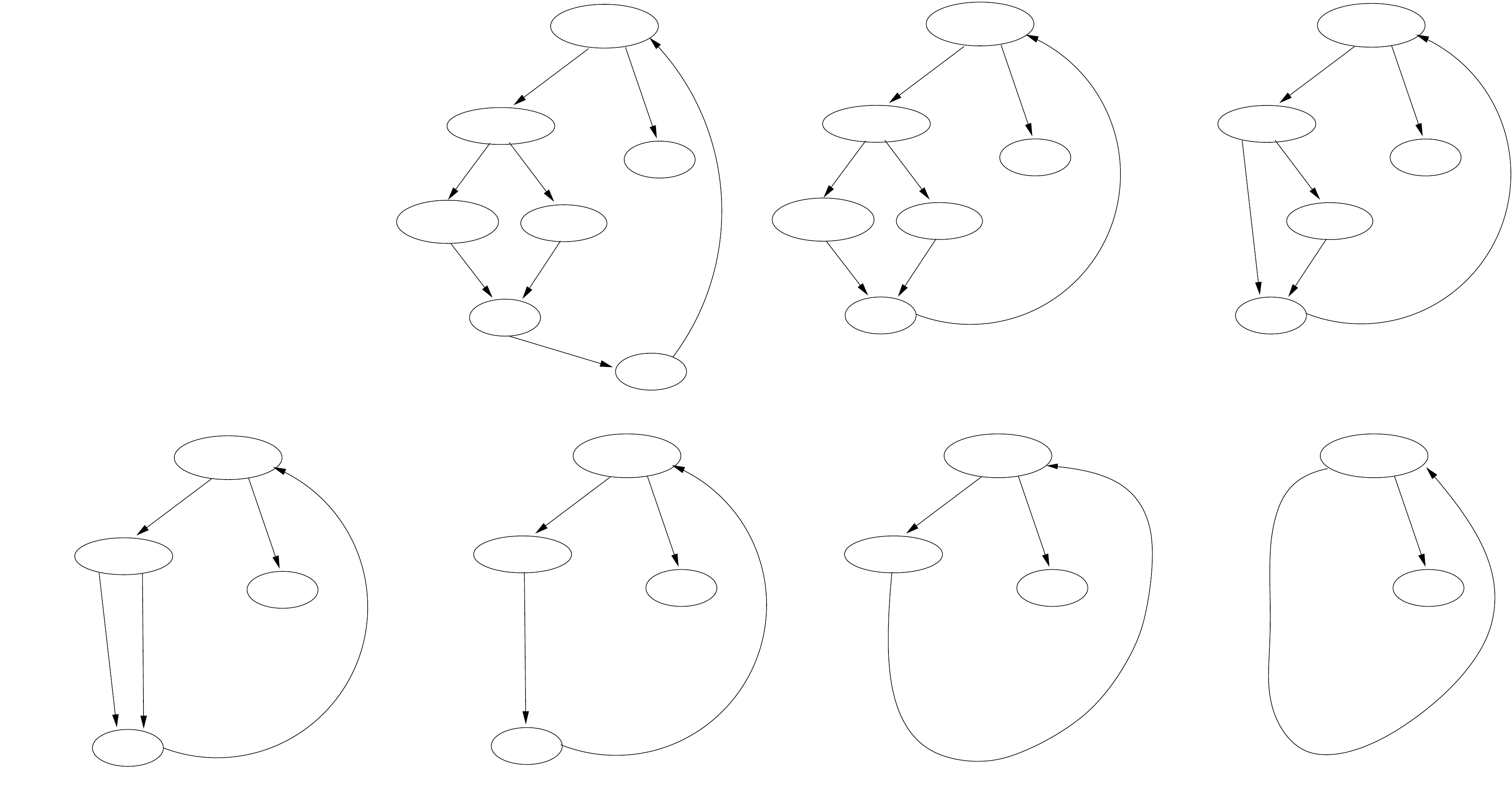
\else
\input{rules_example2.pstex_t}
\fi}
   \caption{CFA Transformation: a) Program, b) CFA, c)--g) Intermediate CFAs, h) Summarization CFA. 
     In the CFAs, $assume(p)$ is represented as $[p]$, 
     $\op_1 \sequence \op_2$ is represented 
     by putting $\op_2$ under $\op_1$,
     and $\op_1 \parallel \op_2$ by putting $\op_2$ beside $\op_1$.
   \label{fig:rules_example2}}
\end{figure*}

In the context of this article, we use the summarization CFA for program analysis,
i.e., we want to verify if an error state of the program is reachable.
The following theorem, which is proved in
Appendix~\ref{sec:appendix_proof}, states that our summarization of a
CFA is correct in this sense. 
 
\begin{theorem}[Correctness of Summarization]\label{thm:summary_correctness}
Let $P = (A, \pci, \pct)$ be a program and
let $A' = (\locs', G')$ be the summarization of~$A$.
Then: 
(i) $\{\pci, \pct\} \subseteq \locs'$, and 
(ii) $\pct$ is reachable in $(A', \pci, \pct)$ if and only if 
$\pct$ is reachable in $P$.
\end{theorem}


The summarization can be performed in polynomial time.
The time taken by Rule 0 is proportional to the number of outgoing edges for $\pct$.
Since each application of Rule 1 or Rule 2 removes at least one edge, 
there can be at most $|G|-1$ such applications. 
A naive way to determine the set of locations and edges to which to apply each rule
requires $O(|V|\cdot k)$ time, where $k$ is the maximum out-degree of locations.
Finally, each application of Rule 2 requires $O(1)$ time, 
and each application of Rule 1 $O(k)$ time.
Therefore, a naive summarization algorithm requires $O(|G|\cdot
|V|\cdot k)$ time, which reduces to $O(|G|\cdot|V|)$ if $k$ is bounded
(i.e., if we rewrite a priori all {\tt switch}es  into nested {\tt if}s).%
\footnote{In our implementation, we use a more efficient algorithm, 
which we do not describe here for lack of space.}

\subsection{LBE versus SBE for Software Model Checking}

The use of LBE instead of the standard SBE
requires no modification to the general model-checking algorithm,
which is still based on the construction of an ART with CEGAR-based refinement.
The main difference is that the LBE has no 
one-to-one correspondence between ART paths and syntactical program paths.
A single CFA edge corresponds to a \emph{set of paths} between its source and target location,
and a single ART path corresponds to a \emph{set of program paths}; 
an ART node represents an overapproximation of the data region 
that is reachable by following \emph{any} of the program paths 
represented by the ART path that leads to it.
This difference leads to two observations.

First, LBE can lead to exponentially-smaller ARTs than SBE, and thus it
can drastically reduce the number successor computations (cf.\ example in
Sect.~\ref{sec-intro})
and the number of abstraction-refinement steps for infeasible error paths.
Each of these operations, however, is typically more expensive than with SBE, 
because the formulas involved are larger and have a more complex structure.

Second, LBE requires a more general representation of abstract states. 
When using SBE, abstract states are typically represented as sets/conjunctions of predicates. 
This is sufficient for practical examples because each abstract state 
represents a data region reachable by a single program path, 
which can be encoded essentially as a conjunction of atomic formulas.
With LBE, such representation would be too coarse,
since each abstract state represents a data region that is reachable on several different program paths.
Therefore, we need to use a representation for arbitrary (and larger) Boolean combinations of predicates.
This generalization of the representation of the abstract state requires 
a generalization of the representation of the transition, i.e.,
the replacement of the Cartesian abstraction 
with a more precise form of abstraction.
In this paper, we evaluate the use of the Boolean abstraction,
which allows for a precise representation of arbitrary Boolean combinations of predicates.

With respect to the traditional SBE approach, 
LBE allows us to trade part of the cost of the \emph{explicit} enumeration of program paths 
with that of the \emph{symbolic} computation of abstract successor states:
rather than having to build large ARTs 
by performing a substantial amount of relatively cheap operations
(Cartesian abstract postoperator applications along single edges
and counterexample analysis of individual program paths),
with LBE we build smaller ARTs by performing more expensive symbolic operations
(Boolean abstract postoperator applications along large portions of the control flow and counterexample analysis of multiple program paths),
involving formulas with a complex Boolean structure.
With LBE, the \emph{cost} of each symbolic operation, rather than their \emph{number},
becomes a critical performance factor.

To this extent, LBE makes it possible to fully exploit the power and
functionality of modern SMT solvers:
First, the capability of modern SMT solvers of performing
large amounts of Boolean 
reasoning allows for handling possibly-big Boolean
combinations of atomic expressions, instead of simple conjunctions.  
Second, the capability of some SMT solvers to perform
All-SMT and interpolant computation (see, e.g.,
\cite{mathsat4}) allows for effectively performing 
SMT-based Boolean abstraction computation \cite{allsmt,bddsmt}
and interpolation-based counterexample analysis \cite{TACAS08} respectively, 
which was shown to outperform previous approaches, 
especially when dealing with complex formulas. 
With SBE, instead, the use of modern SMT technology does not lead to
significant improvements of the whole ART-based algorithm,  
because each SMT query involves simple (and often small) conjunctions 
only.

\section{Performance Evaluation\label{sec-experiments}}

\smallsec{Implementation}
In order to evaluate the proposed verification method,
we integrate our algorithm as a new component into the 
configurable software verification toolkit \cpachecker~\cite{CPACHECKER}.
This implementation is written in {\sc Java}.
All example programs are preprocessed 
and transformed into the simple intermediate language~\cil~\cite{CIL}.
For parsing C~programs, \cpachecker uses a library from the
Eclipse C/C++ Development Kit.
For efficient querying of formulas in the quantifier-free theory of 
rational linear arithmetic and equality with uninterpreted function symbols,
we leverage the SMT solver \mathsat~\cite{mathsat4},
which is integrated as a library (written in C++). 
We use binary decision diagrams (BDDs) for the representation of abstract-state formulas.

We run all experiments on 
a 1.8\,GHz Intel Core2 machine with 2\,GB of RAM and 2\,MB of cache, 
running GNU/Linux. 
We used a timeout of 1\,800\,s and a memory limit of 1.8\,GB.

\smallsec{Example Programs}
We use three categories of benchmark programs.
First, we experiment with programs that
are specifically designed to cause 
an exponential blowup of the ART when using SBE
(\smalltt{test\_locks*}, in the style of the example in Sect.~\ref{sec-intro}).
Second, we use the device-driver programs that were previously
used as benchmarks in the \blast project.
Third, we solve various verification problems for
the SSH client and server software (\smalltt{s3\_clnt*} and \smalltt{s3\_srvr*}),
which share the same program logic,
but check different safety properties.
The safety property is encoded as conditional calls of
a failure location and therefore reduces to the reachability
of a certain error location.
All benchmarks programs from the \blast web page are preprocessed with \cil. 
For the second and third groups of programs, 
we also performed experiments with artificial defects introduced. 

\smallsec{Experimental Configurations}
For a careful and fair performance comparison, we run experiments
on three different configurations. 
First, we use \blast, version 2.5,
which is a highly optimized state-of-the-art software model checker.
\blast is implemented in the programming language {\sc OCaml}.
We run \blast using all four combinations of
breadth-first search (\smalltt{-bfs}) versus depth-first search (\smalltt{-dfs}),
both with and without heuristics for improving the predicate discovery. 
\blast provides five different levels of heuristics
for predicate discovery, and we use only the lowest (\smalltt{-predH 0})
and the highest option (\smalltt{-predH 7}).
Interestingly, every combination is best for some particular example programs,
with considerable differences in runtime and memory consumption.
The configuration using \smalltt{-dfs -predH 7}
is the winner (in terms of solved problems and total runtime)
for the programs without defects, but is not able 
to verify four example programs (timeout).
In the performance table, we provide results obtained using this configuration (column \smalltt{-dfs -predH 7}),
and also the best result among the four configurations for every single
instance (column \smalltt{best result}).
For the unsafe programs, \smalltt{-bfs -predH~7} performs best.
All four configurations use the command-line options 
\smalltt{-craig 2 -nosimplemem -alias ""},
which specify that \blast runs with lazy, Craig-interpolation-based refinement,
no \cil{} preprocessing for memory access, and without pointer analysis.
In all experiments with \blast, we use the same interpolation procedure (\mathsat)
as in our \cpachecker{}-based implementation.
(The results of all four configurations are provided in Appendix~\ref{sec:appendix_results},
 to the reviewers.)

Second, in order to separate the optimization efforts in \blast from 
the conceptual essence of the traditional lazy abstraction algorithm, 
we developed a re-implementation of the traditional algorithms as 
described in the \blast tool article~\cite{BLAST}.
This re-implementation is integrated as component into \cpachecker,
so that the difference between SBE and LBE is only in the algorithms,
not in the environment
(same parser, same BDD package, same query optimization, etc.).
Our SBE implementation uses a DFS algorithm.
This column is labeled as SBE.

Third, we run the experiments using our new LBE algorithm, 
which is also implemented within \cpachecker.
Our LBE implementation uses a DFS algorithm.
This column is labeled as LBE.
Note that the purpose of our experiments is to give evidence
of the performance difference between SBE and LBE, because these two settings
are closest to each other, 
since SBE and LBE differ only in the 
CFA summarization and Boolean abstraction.
The other two columns are provided to give evidence that
the new approach beats the highly optimized traditional implementation \blast.

We actually configured and ran experiments with all four combinations: SBE versus LBE, 
and Cartesian versus Boolean abstraction.
The experimentation clearly showed that SBE does not benefit from
Boolean abstraction in terms of precision, with substantial degrade in
performance: the only programs for which it terminated successfully
were the first five instances of the \smalltt{test\_locks} group.
Similarly, the combination of LBE with Cartesian abstraction fails to
solve any of the experiments, 
due to loss of precision.
Thus, we report only on the two successful
configurations, i.e., SBE in combination with Cartesian abstraction,
and LBE in combination with Boolean abstraction.

%
\begin{table*}[t]
\caption{Performance results
  \label{tab:eval}
}
\vspace{1mm}
\centering
\begin{smallerfont}
\begin{tabular}{lrr|rr}
\hline
              & \multicolumn{2}{c|}{\bf \blast} & \multicolumn{2}{c}{\bf \cpachecker} \\
{\bf Program} & {\tt\scriptsize (best result)} & {\tt\scriptsize (-dfs -predH 7)} & {\bf SBE} & {\bf LBE}\\
\hline
test\_locks\_5.c & 4.50 & 4.96 & 4.01 & {\bf 0.29} \\
test\_locks\_6.c & 7.81 & 8.81 & 7.22 & {\bf 0.32} \\
test\_locks\_7.c & 13.91 & 15.15 & 12.63 & {\bf 0.34} \\
test\_locks\_8.c & 25.00 & 26.49 & 23.93 & {\bf 0.57} \\
test\_locks\_9.c & 46.84 & 49.29 & 52.04 & {\bf 0.38} \\
test\_locks\_10.c & 94.57 & 97.85 & 131.39 & {\bf 0.40} \\
test\_locks\_11.c & 204.55 & 208.78 & MO & {\bf 0.70} \\
test\_locks\_12.c & 529.16 & 533.97 & MO & {\bf 0.46} \\
test\_locks\_13.c & 1229.27 & 1232.87 & MO & {\bf 0.49} \\
test\_locks\_14.c & $>$1800.00 & $>$1800.00 & MO & {\bf 0.50} \\
test\_locks\_15.c & $>$1800.00 & $>$1800.00 & MO & {\bf 0.56} \\
\hline
cdaudio.i.cil.c & 175.76 & 264.12 & MO & {\bf 53.55} \\
diskperf.i.cil.c & $>$1800.00 & $>$1800.00 & MO & {\bf 232.00} \\
floppy.i.cil.c & 218.26 & $>$1800.00 & MO & {\bf 56.36} \\
kbfiltr.i.cil.c & 23.55 & 32.80 & 41.12 & {\bf 7.82} \\
parport.i.cil.c & 738.82 & 915.79 & MO & {\bf 378.04} \\
\hline
s3\_clnt.blast.01.i.cil.c & 33.01 & 1000.41 & 755.81 & {\bf 19.51} \\
s3\_clnt.blast.02.i.cil.c & 62.65 & 312.77 & 1075.45 & {\bf 16.00} \\
s3\_clnt.blast.03.i.cil.c & 60.62 & 314.74 & 746.31 & {\bf 49.50} \\
s3\_clnt.blast.04.i.cil.c & 63.96 & 197.65 & 730.80 & {\bf 25.45} \\
s3\_srvr.blast.01.i.cil.c & 811.27 & 1036.89 & $>$1800.00 & {\bf 125.33} \\
s3\_srvr.blast.02.i.cil.c & 360.47 & 360.47 & $>$1800.00 & {\bf 122.83} \\
s3\_srvr.blast.03.i.cil.c & 276.19 & 276.19 & $>$1800.00 & {\bf 98.47} \\
s3\_srvr.blast.04.i.cil.c & 175.64 & 301.85 & $>$1800.00 & {\bf 71.77} \\
s3\_srvr.blast.06.i.cil.c & 304.63 & 304.63 & $>$1800.00 & {\bf 59.70} \\
s3\_srvr.blast.07.i.cil.c & 478.05 & 666.53 & $>$1800.00 & {\bf 85.82} \\
s3\_srvr.blast.08.i.cil.c & 115.76 & 115.76 & $>$1800.00 & {\bf 61.29} \\
s3\_srvr.blast.09.i.cil.c & 445.21 & 1037.09 & $>$1800.00 & {\bf 126.47} \\
s3\_srvr.blast.10.i.cil.c & 115.10 & 115.10 & $>$1800.00 & {\bf 63.36} \\
s3\_srvr.blast.11.i.cil.c & 367.98 & 844.28 & $>$1800.00 & {\bf 162.76} \\
s3\_srvr.blast.12.i.cil.c & 304.05 & 304.05 & $>$1800.00 & {\bf 170.33} \\
s3\_srvr.blast.13.i.cil.c & 580.33 & 878.54 & $>$1800.00 & {\bf 74.49} \\
s3\_srvr.blast.14.i.cil.c & 303.21 & 303.21 & $>$1800.00 & {\bf 50.38} \\
s3\_srvr.blast.15.i.cil.c & 115.88 & 115.88 & $>$1800.00 & {\bf 21.01} \\
s3\_srvr.blast.16.i.cil.c & 305.11 & 305.11 & $>$1800.00 & {\bf 127.82} \\
\hline
{\bf TOTAL (solved/time)}\!\! & {\bf ~32\,/\,8591.12} & {\bf ~31\,/\,12182.03} & {\bf ~11\,/\,3580.71} & {\bf ~35\,/\,2265.07}\\
\hline
{\bf TOTAL w/o {\tt test\_locks*}} & {\bf ~23\,/\,6435.51} & {\bf ~22\,/\,10003.06} & {\bf ~5\,/\,3349.48} & {\bf ~24\,/\,2260.07}\\
\hline
\end{tabular}
\end{smallerfont}
\end{table*}

\begin{table*}[t]
\caption{Detailed comparison between SBE and LBE encodings;
 entries marked with (*) denote partial statistics 
 for analyses that terminated unsuccessfully
 (if available).
  \label{tab:eval_explicit_vs_summary_stats}
}
\vspace{-2mm}
\centering
\begin{smallerfont}
\begin{tabular}{l@{\hspace{3mm}}|r@{\hspace{3mm}}r@{\hspace{1mm}}@{\hspace{1mm}}rrr|r@{\hspace{3mm}}r@{\hspace{1mm}}@{\hspace{1mm}}rrr}
\hline
 & \multicolumn{5}{c|}{\bf SBE} & 
   \multicolumn{5}{|c}{\bf LBE} \\
& {\bf ART} & {\bf \# ref} & \multicolumn{3}{c|}{\bf \# predicates}
& {\bf ART} & {\bf \# ref} & \multicolumn{3}{c}{\bf \# predicates} \\
{\bf Program} & {\bf size} & {\bf steps} & {\bf Tot} & {\bf Avg} & {\bf Max} &
                 {\bf size} & {\bf steps} & {\bf Tot} & {\bf Avg} & {\bf Max} \\
\hline
test\_locks\_5.c & 1344 & 50 & 10 & 3 & 10 & 4 & 0 & 0 & 0 & 0 \\
test\_locks\_6.c & 2301 & 72 & 12 & 4 & 12 & 4 & 0 & 0 & 0 & 0 \\
test\_locks\_7.c & 3845 & 98 & 14 & 5 & 14 & 4 & 0 & 0 & 0 & 0 \\
test\_locks\_8.c & 6426 & 128 & 16 & 6 & 16 & 4 & 0 & 0 & 0 & 0 \\
test\_locks\_9.c & 10926 & 162 & 18 & 7 & 18 & 4 & 0 & 0 & 0 & 0 \\
test\_locks\_10.c & 19091 & 200 & 20 & 8 & 20 & 4 & 0 & 0 & 0 & 0 \\
test\_locks\_11.c & 24779(*) & 242(*) & 22(*) & 9(*) & 22(*) & 4 & 0 & 0 & 0 & 0 \\
test\_locks\_12.c & 28119(*) & 288(*) & 24(*) & 10(*) & 24(*) & 4 & 0 & 0 & 0 & 0 \\
test\_locks\_13.c & 31739(*) & 338(*) & 26(*) & 10(*) & 26(*) & 4 & 0 & 0 & 0 & 0 \\
test\_locks\_14.c & 35178(*) & 392(*) & 28(*) & 11(*) & 28(*) & 4 & 0 & 0 & 0 & 0 \\
test\_locks\_15.c & 38777(*) & 450(*) & 30(*) & 12(*) & 30(*) & 4 & 0 & 0 & 0 & 0 \\
\hline
cdaudio.i.cil.c & 53323(*) & 445(*) & 147(*) & 9(*) & 78(*) & 6909 & 140 & 79 & 5 & 16 \\
diskperf.i.cil.c & -- & -- & -- & -- & -- & 4890 & 145 & 56 & 6 & 21 \\
floppy.i.cil.c & 31079(*) & 301(*) & 79(*) & 7(*) & 35(*) & 9668 & 176 & 58 & 4 & 13 \\
kbfiltr.i.cil.c & 19640 & 153 & 53 & 5 & 27 & 1577 & 47 & 18 & 2 & 6 \\
parport.i.cil.c & 26188(*) & 360(*) & 143(*) & 4(*) & 41(*) & 38488 & 474 & 168 & 4 & 17 \\
\hline
s3\_clnt.blast.01.i.cil.c & 122678 & 557 & 59 & 20 & 59 & 36 & 5 & 47 & 11 & 47 \\
s3\_clnt.blast.02.i.cil.c & 354132 & 532 & 55 & 19 & 55 & 36 & 5 & 51 & 12 & 51 \\
s3\_clnt.blast.03.i.cil.c & 196599 & 534 & 55 & 19 & 55 & 39 & 5 & 75 & 18 & 75 \\
s3\_clnt.blast.04.i.cil.c & 172444 & 538 & 55 & 19 & 55 & 36 & 5 & 47 & 11 & 47 \\
s3\_srvr.blast.01.i.cil.c & 232195(*) & 774(*) & 70(*) & 20(*) & 70(*) & 101 & 6 & 88 & 22 & 88 \\
s3\_srvr.blast.02.i.cil.c & 254667(*) & 745(*) & 79(*) & 19(*) & 78(*) & 109 & 7 & 75 & 18 & 75 \\
s3\_srvr.blast.03.i.cil.c & -- & -- & -- & -- & -- & 91 & 6 & 85 & 21 & 85 \\
s3\_srvr.blast.04.i.cil.c & -- & -- & -- & -- & -- & 103 & 7 & 82 & 20 & 82 \\
s3\_srvr.blast.06.i.cil.c & 295698(*) & 576(*) & 63(*) & 14(*) & 63(*) & 94 & 6 & 84 & 21 & 84 \\
s3\_srvr.blast.07.i.cil.c & -- & -- & -- & -- & -- & 92 & 5 & 85 & 21 & 85 \\
s3\_srvr.blast.08.i.cil.c & 279991(*) & 549(*) & 57(*) & 15(*) & 57(*) & 89 & 5 & 88 & 22 & 88 \\
s3\_srvr.blast.09.i.cil.c & 189541(*) & 720(*) & 72(*) & 16(*) & 71(*) & 193 & 4 & 72 & 18 & 72 \\
s3\_srvr.blast.10.i.cil.c & 307671(*) & 597(*) & 55(*) & 16(*) & 55(*) & 91 & 5 & 79 & 19 & 79 \\
s3\_srvr.blast.11.i.cil.c & -- & -- & -- & -- & -- & 48 & 6 & 69 & 17 & 69 \\
s3\_srvr.blast.12.i.cil.c & 258546(*) & 563(*) & 57(*) & 15(*) & 57(*) & 99 & 6 & 94 & 23 & 94 \\
s3\_srvr.blast.13.i.cil.c & 167333(*) & 682(*) & 70(*) & 18(*) & 69(*) & 90 & 5 & 81 & 20 & 81 \\
s3\_srvr.blast.14.i.cil.c & 318982(*) & 643(*) & 65(*) & 13(*) & 64(*) & 92 & 6 & 83 & 20 & 83 \\
s3\_srvr.blast.15.i.cil.c & 279319(*) & 579(*) & 58(*) & 15(*) & 58(*) & 71 & 4 & 71 & 17 & 71 \\
s3\_srvr.blast.16.i.cil.c & 346185(*) & 596(*) & 59(*) & 12(*) & 58(*) & 98 & 6 & 86 & 21 & 86 \\
\hline
\end{tabular}
\end{smallerfont}
\end{table*}

\begin{table}
\caption{Performance results, 
  programs with artificial bugs.
  \label{tab:eval_bug}
}
\vspace{1mm}
\centering
\begin{smallerfont}
\begin{tabular}{lr|rr}
\hline
              & {\bf \blast} & \multicolumn{2}{c}{\bf \cpachecker} \\
\multicolumn{2}{l|}{{\bf Program} \hfill {\tt\scriptsize (best result)}} & {\bf SBE} & {\bf LBE}\\
\hline
cdaudio.BUG.i.cil.c & 18.79 & 74.39 & {\bf 9.85} \\
diskperf.BUG.i.cil.c & 889.79 & 26.53 & {\bf 6.78} \\
floppy.BUG.i.cil.c & 119.60 & 36.49 & {\bf 4.30} \\
kbfiltr.BUG.i.cil.c & 46.80 & 75.45 & {\bf 11.52} \\
parport.BUG.i.cil.c & {\bf 1.67} & 14.62 & 2.64 \\
\hline
s3\_clnt.blast.01.BUG.i.cil.c & 8.84 & 1514.90 & {\bf 3.33} \\
s3\_clnt.blast.02.BUG.i.cil.c & 9.02 & 843.42 & {\bf 3.27} \\
s3\_clnt.blast.03.BUG.i.cil.c & 6.64 & 780.72 & {\bf 2.61} \\
s3\_clnt.blast.04.BUG.i.cil.c & 9.78 & 724.04 & {\bf 3.18} \\
s3\_srvr.blast.01.BUG.i.cil.c & 7.59 & MO & {\bf 2.09} \\
s3\_srvr.blast.02.BUG.i.cil.c & 7.16 & $>$1800.00 & {\bf 2.10} \\
s3\_srvr.blast.03.BUG.i.cil.c & 7.42 & $>$1800.00 & {\bf 2.08} \\
s3\_srvr.blast.04.BUG.i.cil.c & 7.33 & $>$1800.00 & {\bf 1.93} \\
s3\_srvr.blast.06.BUG.i.cil.c & 39.81 & MO & {\bf 5.08} \\
s3\_srvr.blast.07.BUG.i.cil.c & 310.84 & $>$1800.00 & {\bf 28.35} \\
s3\_srvr.blast.08.BUG.i.cil.c & 40.51 & $>$1800.00 & {\bf 36.47} \\
s3\_srvr.blast.09.BUG.i.cil.c & 265.48 & $>$1800.00 & {\bf 4.94} \\
s3\_srvr.blast.10.BUG.i.cil.c & 40.24 & $>$1800.00 & {\bf 12.01} \\
s3\_srvr.blast.11.BUG.i.cil.c & 49.05 & $>$1800.00 & {\bf 4.80} \\
s3\_srvr.blast.12.BUG.i.cil.c & 38.66 & $>$1800.00 & {\bf 6.11} \\
s3\_srvr.blast.13.BUG.i.cil.c & 251.56 & $>$1800.00 & {\bf 15.20} \\
s3\_srvr.blast.14.BUG.i.cil.c & 39.94 & 1656.54 & {\bf 4.63} \\
s3\_srvr.blast.15.BUG.i.cil.c & 40.19 & $>$1800.00 & {\bf 10.19} \\
s3\_srvr.blast.16.BUG.i.cil.c & 39.54 & $>$1800.00 & {\bf 5.21} \\
\hline
{\bf TOTAL (solved/time)} & {\bf 24\,/\,2296.25} & {\bf 10\,/\,5747.10} & {\bf 24\,/\,188.67}\\
\hline
\end{tabular}
\end{smallerfont}
\end{table}

\smallsec{Discussion of Evaluation Results}
Tables~\ref{tab:eval} and~\ref{tab:eval_bug} present performance results of our experiments,
for the safe and unsafe programs respectively.
All runtimes are given in seconds of processor time, `$>$1800.00'
indicates a timeout, `MO' indicates an out-of-memory.
Table~\ref{tab:eval_explicit_vs_summary_stats} 
shows statistics about the algorithms for SBE and LBE only.

The first group of experiments in Table~\ref{tab:eval} 
shows that the time complexity of SBE (and \blast)
can grow exponentially in the number of nested conditional
statements, as expected.
Table~\ref{tab:eval_explicit_vs_summary_stats} explains why
the SBE approach exhausts the memory:
the number of abstract nodes in the reachability tree grows exponentially
in the number of nested conditional statements.
Therefore, this approach does not scale.
The LBE approach reduces the loop-free part of the branching control-flow structure
to a few edges (cf.~example in the introduction),
and the size of the ART is constant, because only the
structure inside the body of the loop changes.
There are no refinement steps necessary in the LBE approach,
because the edges to the error location are infeasible.
Therefore, no predicates are used.
The runtime of the LBE approach slightly increases
with the size of the program, because the formulas
that are sent to the SMT solver are slightly increasing.
Although in principle the complexity of the SMT problem grows exponentially in the size of the formulas,
the heuristics used by SMT solvers avoid the exponential enumeration
that we observe in the case of SBE.

For the two other classes of experiments, we see that LBE
is able to successfully complete all benchmarks,
and shows significant performance gains over SBE.
SBE is able to solve only about
one third of all benchmarks, and for the ones that complete,
it is clearly outperformed by LBE.
In Table~\ref{tab:eval_explicit_vs_summary_stats}, we see that SBE
has in general a much larger ART.
In Table~\ref{tab:eval} we observe not only that LBE performs 
significantly better than the \smalltt{-dfs -predH 7} configuration of \blast, 
but that LBE is better than any \blast configuration (column \smalltt{best result}).
LBE performed best also in finding the error paths (cf. Table~\ref{tab:eval_bug}), 
clearly outperforming both SBE and \blast.

In summary, the experiments show that the LBE approach outperforms the
SBE approach, both for correct and defective programs.
This provides evidence of the benefits of a ``more symbolic''
analysis as performed in the LBE approach.
One might argue that our \cpachecker-based SBE implementation might be sub-optimal
although it uses the same implementation and execution environment as LBE;
this is why we compare with \blast as well,
and the experiments become even more impressive when considering that \blast is the result of
several years of fine-tuning.

\section{Conclusion and Future Work}
\label{sec-conclusion}

We have proposed LBE as an alternative to the SBE
model-checking approach, based on the idea that
transitions in the abstract space represent larger fragments of the
program. 
Our novel approach results in significantly smaller ARTs, 
where abstract successor computations are more involved,
and thus trading cost of many explicit enumerations of program paths  
with the cost of symbolic successor computations.
A thorough
experimental evaluation shows the advantages of LBE against both our
implementation of SBE and the state-of-the-art \blast{} system.

In our future work, we plan to
implement McMillan's interpolation-based lazy-abstraction approach~\cite{McMillanCAV06},
and experiment with SBE versus LBE versions of his algorithm.
Furthermore, we plan to investigate the use of adjustable precision-based 
techniques for the construction of the large blocks 
on-the-fly (instead of the current preprocessing step).
This would enable a dynamic adjustment of the size of the large blocks,
and thus we could fine-tune the amount of work that is delegated to the SMT solver.
Also, we plan to explore other techniques for computing abstract successors 
which are more precise than Cartesian abstraction but less expensive than Boolean abstraction.

\smallsec{Acknowledgments}
We thank Roman Manevich for interesting discussions about \blast's performance bottlenecks.

\appendix
\section{Appendix}
\label{sec:appendix}

\subsection{Proof of Theorem~\ref{thm:summary_correctness}}
\label{sec:appendix_proof}
In order to prove Theorem~\ref{thm:summary_correctness},
we introduce some auxiliary lemmas.

\begin{lemma}\label{thm:lemma_sp_disjunction}
  Let $(l, \op, l')$ be a CFA edge,
  and $\{\varphi_i\}_i$ a collection of formulas.
  Then 
  \begin{displaymath}
    \SP_{op}(\textstyle\bigvee_i \varphi_i) \equiv \textstyle\bigvee_i \SP_{op}(\varphi_i).
  \end{displaymath}
\end{lemma}
\begin{proof}
  \newcommand{\bvee}{\textstyle\bigvee}
  If $\op$ is an assignment operation $s := e$, then
  \begin{displaymath}
    \newcommand{\sh}[1]{[#1 \mapsto \widehat{#1}]}
    \begin{split}
      \SP_{s := e}(\bvee_i \varphi_i) 
      & ~=~ \exists\,\widehat{s}.
      ((\bvee_i \varphi_i)_{\sh{s}} \land (s = e_{\sh{s}})) \\
      & ~\equiv~ \exists\,\widehat{s}.
      (\bvee_i ({\varphi_i}_{\sh{s}} \land (s = e_{\sh{s}}))) \\
      & ~\equiv~ \bvee_i(\exists\,\widehat{s}.
      ({\varphi_i}_{\sh{s}} \land (s = e_{\sh{s}}))) \\
      & ~\equiv~ \bvee_i \SP_{s := e}(\varphi_i)
    \end{split}
  \end{displaymath}

  \smallskip \noindent
  If $\op$ is an assume operation $\mathit{assume}(p)$, then
  \begin{displaymath}
    \begin{split}
      \SP_{\mathit{assume}(p)}(\bvee_i \varphi_i) 
      & ~=~ (\bvee_i \varphi_i) \land p \\
      & ~\equiv~ \bvee_i (\varphi_i \land p) \\
      & ~\equiv~ \bvee_i \SP_{\mathit{assume}(p)}(\varphi_i)
    \end{split}
  \end{displaymath}

  \smallskip \noindent
  The remaining two cases can be proven by induction.\\
  If $\op = \op_1 \sequence \op_2$, then
  \begin{displaymath}
    \begin{split}
      \SP_{\op_1\sequence \op_2}(\bvee_i \phi_i)
      & ~=~ \SP_{\op_2}(\SP_{\op_1}(\bvee_i \phi_i)) \\
      & ~\equiv~ \SP_{\op_2}(\bvee_i \SP_{\op_1}(\phi_i)) \\
      & ~\equiv~ \bvee_i \SP_{\op_2}(\SP_{\op_1}(\phi_i)) \\
      & ~\equiv~ \bvee_i \SP_{\op_1\sequence \op_2}(\phi_i)
    \end{split}
  \end{displaymath}

  \smallskip \noindent
  If $\op = \op_1 \parallel \op_2$, then
  \begin{displaymath}
    \begin{split}
      \SP_{\op_1\parallel \op_2}(\bvee_i \phi_i)
      & ~=~ \SP_{\op_1}(\bvee_i \phi_i) \lor \SP_{\op_2}(\bvee_i \phi_i) \\
      & ~\equiv~ (\bvee_i \SP_{\op_1}(\phi_i)) \lor (\bvee_i \SP_{\op_2}(\phi_i)) \\
      & ~\equiv~ \bvee_i (\SP_{\op_1}(\phi_i) \lor \SP_{\op_2}(\phi_i)) \\
      & ~\equiv~ \bvee_i \SP_{\op_1\parallel \op_2}(\phi_i)
    \end{split}
  \end{displaymath}
\qed
\end{proof}

\begin{lemma}\label{thm:lemma_summary_sound}
  Let $A = (L, G)$ be a CFA, 
  and let $A' = (L', G')$ be a summarization of $A$.
  Let $\path$ be a path in $A$ such that 
  its initial and final locations occur also in $L'$.
  Then for all $\varphi$,
  there exists a path $\path'$ in $A'$,
  with the same initial and final locations as $\path$,
  such that $\SP_\path(\varphi) \models \SP_{\path'}(\varphi)$.
\end{lemma}

\begin{proof}
  CFA $A'$ is obtained from $A$ by a sequence of $n$ rule applications.
  If $n=0$ we have $A' = A$.
  If the lemma holds for one rule application, we can show by induction
  that the lemma holds for any finite sequence of rule applications.

  We now show that the lemma holds for one rule application.
  Let \mbox{$\path = \path_1,(l_i, \op_i, l_j)$}. 
  The proof is by induction on the length of $\path$.
  (The base case is when $\path_1$ is empty.)

  If $l_i \in L'$, by the inductive hypothesis 
  there exists a path $\path'_1$ in $A'$ such that 
  $\SP_{\path_1}(\varphi) \models \SP_{\path'_1}(\varphi)$.
  If $(l_i, \op_i, l_j) \in G'$, then we can take 
  \mbox{$\path' = \path'_1, (l_i, \op_i, l_j)$}.
  Otherwise, $(l_i, \op_i, l_j)$ must have been removed 
  by an application of Rule 2,
  \footnote{It could not have been removed by Rule 1, 
    because when Rule 1 removes the edges 
    $(\cdot, \cdot, l)$ and $(l, \cdot, \cdot)$, 
    it removes also the location $l$.}
  and so $G'$ contains an edge $(l_i, \op_i \parallel \cdot, l_j)$. 
  Therefore, we can take 
  \mbox{$\path' = \path'_1, (l_i, \op_i \parallel \cdot, l_j)$}.
  
  If $l_i \not\in L'$, then by hypothesis
  $\path \equiv \path_2, (l_k, \op_k, l_i), (l_i, \op_i, l_j)$.
  Moreover, $l_i$ has been removed by an application of Rule 1.
  By the definition of Rule 1, $(l_k, \op_k, l_i)$ is the only incoming edge for $l_i$ in $G$.
  Therefore, $G'$ contains an edge
  $(l_k, \op_k \sequence \op_i, l_j)$ and clearly $l_k \in L'$.
  Thus, by the inductive hypothesis 
  there exists a path $\path'_2$ in $A'$ such that 
  $\SP_{\path_2}(\varphi) \models \SP_{\path'_2}(\varphi)$,
  and so we can take \mbox{$\path' = \path'_2, (l_k, \op_k \sequence \op_i, l_j)$}.
  \qed
\end{proof}

\begin{lemma}\label{thm:lemma_summary_compl}
  Let $A = (L, G)$ be a CFA, 
  and let $A' = (L', G')$ be a summarization of $A$.
  Let $\path'$ be a path in $A'$.
  Then for all $\varphi$, 
  there exists a set $\Sigma$ of paths in $A$,
  with the same initial and final locations as $\path'$,
  such that 
  $\SP_{\path'}(\varphi) \equiv \bigvee_{\path \in \Sigma} \SP_\path(\varphi)$.
\end{lemma}

\begin{proof}
  CFA $A'$ is obtained from $A$ by a sequence of $n$ rule applications.
  If $n=0$ we have $A' = A$.
  If the lemma holds for one rule application, we can show by induction
  that the lemma holds for any finite sequence of rule applications.

  We now show that the lemma holds for one rule application.
  Let \mbox{$\path' = \path'_p,(l_i, \op_i, l_j)$} be a path in $A'$. 
  The proof is by induction on the length of $\path'$.
  (The base case is when $\path'_p$ is empty.)

  First, we observe that all locations in $\path'$ occur also in $G$.

  By the inductive hypothesis,
  there exists a set $\Sigma_p$ of paths in $A$,
  with the same initial and final locations as $\path'_p$,
  such that 
  $\SP_{\path'_p}(\varphi) \equiv \bigvee_{\path_p \in \Sigma_p} \SP_{\path_p}(\varphi)$.

  If $(l_i, \op_i, l_j) \in G$, then we can take 
  $\Sigma = \{ \path_p, (l_i, \op_i, l_j)~|~\path_p \in \Sigma_p \}$ 
  (by Lemma~\ref{thm:lemma_sp_disjunction}).

  Otherwise, $(l_i, \op_i, l_j)$ was generated 
  by an application of one of the Rules.
  If it was generated by Rule 1, 
  then $G$ contains two edges 
  $(l_i, \op'_i, l_k)$ and $(l_k, \op_k, l_j)$ 
  such that $\op_i = \op'_i \sequence \op_k$.
  Then we can take 
  \mbox{$\Sigma = \{ \path_p, (l_i, \op'_i, l_k), (l_k, \op_k, l_j)~|~\path_p \in \Sigma_p\}$}
  (by Lemma~\ref{thm:lemma_sp_disjunction}).
  If $(l_i, \op_i, l_j)$ was generated by Rule 2,
  then $G$ contains two edges
  $(l_i, \op'_i, l_j)$ and $(l_i, \op''_i, l_j)$
  such that $\op_i = \op'_i \parallel \op''_i$.
  Let \mbox{$\Sigma_1 = \{ \path_p, (l_i, \op'_i, l_j) ~|~ \path_p \in \Sigma_p \}$} 
  and
  \mbox{$\Sigma_2 = \{ \path_p, (l_i, \op''_i, l_j) ~|~ \path_p \in \Sigma_p \}$}.
  Then we can take \mbox{$\Sigma = \Sigma_1 \cup \Sigma_2$} 
  (by Lemma~\ref{thm:lemma_sp_disjunction}).
  
  \qed
\end{proof}

\begin{proof}
Now we prove Theorem~\ref{thm:summary_correctness}.

\begin{itemize}
\item[(i)]
  The only Rule that removes locations is Rule 1. 
  Since $\pci$ has no incoming edges (by definition) 
  and $\pct$ has no outgoing edges (because of Rule 0), 
  they cannot be removed by Rule 1.

\item[(ii)]
  ``$\rightarrow$'' 
  Follows from Lemma~\ref{thm:lemma_summary_sound} 
  and (i).\\
  ``$\leftarrow$''
  Follows from Lemma~\ref{thm:lemma_summary_compl} and (i).
\end{itemize}
\qed
\end{proof}

\subsection{Comparison among different \blast configurations}
\label{sec:appendix_results}

\begin{table*}
\caption{Comparison among different configurations of \blast.
(NP indicates 'no new predicates found during refinement'.)
  \label{tab:eval_blast}
}
\centering
\begin{smallerfont}
  \begin{tabular}{l@{\hspace{-4mm}}rrrrr}
    \hline
                   & {\bf \blast 1} & {\bf \blast 2} & {\bf \blast 3} & {\bf \blast 4} & {\bf \blast B}\\
    {\bf Program}  & {\tt\scriptsize (-bfs -predH 0)} & {\tt\scriptsize (-bfs -predH 7)} 
                   & {\tt\scriptsize (-dfs -predH 0)} & {\tt\scriptsize (-dfs -predH 7)} & {\tt\scriptsize (best result)}\\
\hline
test\_locks\_5.c & 8.36 & 8.40 & {\bf 4.50} & 4.96 & {\bf 4.50} \\
test\_locks\_6.c & 17.63 & 17.29 & {\bf 7.81} & 8.81 & {\bf 7.81} \\
test\_locks\_7.c & 39.90 & 37.83 & {\bf 13.91} & 15.15 & {\bf 13.91} \\
test\_locks\_8.c & 86.98 & 86.69 & {\bf 25.00} & 26.49 & {\bf 25.00} \\
test\_locks\_9.c & 173.63 & 189.96 & {\bf 46.84} & 49.29 & {\bf 46.84} \\
test\_locks\_10.c & 500.30 & 483.07 & {\bf 94.57} & 97.85 & {\bf 94.57} \\
test\_locks\_11.c & 1645.90 & 1534.20 & {\bf 204.55} & 208.78 & {\bf 204.55} \\
test\_locks\_12.c & $>$1800.00 & $>$1800.00 & {\bf 529.16} & 533.97 & {\bf 529.16} \\
test\_locks\_13.c & $>$1800.00 & $>$1800.00 & {\bf 1229.27} & 1232.87 & {\bf 1229.27} \\
test\_locks\_14.c & $>$1800.00 & $>$1800.00 & $>$1800.00 & $>$1800.00 & $>$1800.00 \\
test\_locks\_15.c & $>$1800.00 & $>$1800.00 & $>$1800.00 & $>$1800.00 & $>$1800.00 \\
cdaudio.i.cil.c & 380.83 & 475.67 & {\bf 175.76} & 264.12 & {\bf 175.76} \\
diskperf.i.cil.c & -- & $>$1800.00 & NP & $>$1800.00 & $>$1800.00 \\
floppy.i.cil.c & {\bf 218.26} & $>$1800.00 & NP & $>$1800.00 & {\bf 218.26} \\
kbfiltr.i.cil.c & {\bf 23.55} & 69.07 & NP & 32.80 & {\bf 23.55} \\
parport.i.cil.c & {\bf 738.82} & $>$1800.00 & NP & 915.79 & {\bf 738.82} \\
s3\_clnt.blast.01.i.cil.c & 72.55 & 526.77 & {\bf 33.01} & 1000.41 & {\bf 33.01} \\
s3\_clnt.blast.02.i.cil.c & 80.57 & 268.67 & {\bf 62.65} & 312.77 & {\bf 62.65} \\
s3\_clnt.blast.03.i.cil.c & 124.99 & 440.25 & {\bf 60.62} & 314.74 & {\bf 60.62} \\
s3\_clnt.blast.04.i.cil.c & 140.60 & 138.75 & {\bf 63.96} & 197.65 & {\bf 63.96} \\
s3\_srvr.blast.01.i.cil.c & 1030.27 & MO & {\bf 811.27} & 1036.89 & {\bf 811.27} \\
s3\_srvr.blast.02.i.cil.c & $>$1800.00 & 811.77 & 1088.43 & {\bf 360.47} & {\bf 360.47} \\
s3\_srvr.blast.03.i.cil.c & 1166.38 & 424.53 & 961.72 & {\bf 276.19} & {\bf 276.19} \\
s3\_srvr.blast.04.i.cil.c & 208.89 & {\bf 175.64} & 1393.08 & 301.85 & {\bf 175.64} \\
s3\_srvr.blast.06.i.cil.c & $>$1800.00 & $>$1800.00 & 653.62 & {\bf 304.63} & {\bf 304.63} \\
s3\_srvr.blast.07.i.cil.c & $>$1800.00 & $>$1800.00 & {\bf 478.05} & 666.53 & {\bf 478.05} \\
s3\_srvr.blast.08.i.cil.c & $>$1800.00 & 411.92 & 647.87 & {\bf 115.76} & {\bf 115.76} \\
s3\_srvr.blast.09.i.cil.c & $>$1800.00 & 1296.56 & {\bf 445.21} & 1037.09 & {\bf 445.21} \\
s3\_srvr.blast.10.i.cil.c & $>$1800.00 & $>$1800.00 & 645.23 & {\bf 115.10} & {\bf 115.10} \\
s3\_srvr.blast.11.i.cil.c & 1692.77 & 1011.15 & {\bf 367.98} & 844.28 & {\bf 367.98} \\
s3\_srvr.blast.12.i.cil.c & $>$1800.00 & 1188.43 & 658.16 & {\bf 304.05} & {\bf 304.05} \\
s3\_srvr.blast.13.i.cil.c & $>$1800.00 & MO & {\bf 580.33} & 878.54 & {\bf 580.33} \\
s3\_srvr.blast.14.i.cil.c & $>$1800.00 & 463.95 & 653.85 & {\bf 303.21} & {\bf 303.21} \\
s3\_srvr.blast.15.i.cil.c & $>$1800.00 & 604.01 & 645.35 & {\bf 115.88} & {\bf 115.88} \\
s3\_srvr.blast.16.i.cil.c & $>$1800.00 & 653.87 & 651.30 & {\bf 305.11} & {\bf 305.11} \\
\hline
{\bf TOTAL (solved/time)} & {\bf ~19\,/\,8351.18} & {\bf ~23\,/\,11318.45} & {\bf ~29\,/\,13233.06} & {\bf ~31\,/\,12182.03} & {\bf ~32\,/\,8591.12}\\
\hline
\end{tabular}
\end{smallerfont}
\end{table*}

\begin{table*}
\caption{Comparison among different configurations of \blast,
  programs with artificial bugs.
(NP indicates 'no new predicates found during refinement'.)  \label{tab:eval_blast_bug}
}
\centering
\begin{smallerfont}
  \begin{tabular}{l@{\hspace{-4mm}}rrrrr}
    \hline
                   & {\bf \blast 1} & {\bf \blast 2} & {\bf \blast 3} & {\bf \blast 4} & {\bf \blast B}\\
    {\bf Program}  & {\tt\scriptsize (-bfs -predH 0)} & {\tt\scriptsize (-bfs -predH 7)} 
                   & {\tt\scriptsize (-dfs -predH 0)} & {\tt\scriptsize (-dfs -predH 7)} & {\tt\scriptsize (best result)}\\
\hline
cdaudio.BUG.i.cil.c & 108.85 & 99.82 & 26.83 & {\bf 18.79} & {\bf 18.79} \\
diskperf.BUG.i.cil.c & {\bf 889.79} & $>$1800.00 & 926.70 & $>$1800.00 & {\bf 889.79} \\
floppy.BUG.i.cil.c & {\bf 119.60} & $>$1800.00 & 127.68 & $>$1800.00 & {\bf 119.60} \\
kbfiltr.BUG.i.cil.c & 70.83 & 144.25 & NP & {\bf 46.80} & {\bf 46.80} \\
parport.BUG.i.cil.c & 5.70 & 10.95 & {\bf 1.67} & 2.24 & {\bf 1.67} \\
s3\_clnt.blast.01.BUG.i.cil.c & 1003.92 & 28.30 & 304.63 & {\bf 8.84} & {\bf 8.84} \\
s3\_clnt.blast.02.BUG.i.cil.c & 118.48 & {\bf 9.02} & 131.42 & 12.26 & {\bf 9.02} \\
s3\_clnt.blast.03.BUG.i.cil.c & 167.73 & {\bf 6.64} & 133.97 & 12.20 & {\bf 6.64} \\
s3\_clnt.blast.04.BUG.i.cil.c & 187.18 & {\bf 9.78} & 139.04 & 11.70 & {\bf 9.78} \\
s3\_srvr.blast.01.BUG.i.cil.c & 103.06 & {\bf 7.59} & $>$1800.00 & 162.90 & {\bf 7.59} \\
s3\_srvr.blast.02.BUG.i.cil.c & 123.00 & {\bf 7.16} & $>$1800.00 & 183.34 & {\bf 7.16} \\
s3\_srvr.blast.03.BUG.i.cil.c & 55.21 & {\bf 7.42} & 1434.01 & 49.74 & {\bf 7.42} \\
s3\_srvr.blast.04.BUG.i.cil.c & 79.16 & {\bf 7.33} & $>$1800.00 & 53.22 & {\bf 7.33} \\
s3\_srvr.blast.06.BUG.i.cil.c & 1623.73 & 56.11 & 558.18 & {\bf 39.81} & {\bf 39.81} \\
s3\_srvr.blast.07.BUG.i.cil.c & 1582.86 & {\bf 310.84} & 1327.50 & MO & {\bf 310.84} \\
s3\_srvr.blast.08.BUG.i.cil.c & $>$1800.00 & 73.59 & 530.10 & {\bf 40.51} & {\bf 40.51} \\
s3\_srvr.blast.09.BUG.i.cil.c & $>$1800.00 & {\bf 265.48} & 1284.77 & MO & {\bf 265.48} \\
s3\_srvr.blast.10.BUG.i.cil.c & $>$1800.00 & 66.88 & 528.29 & {\bf 40.24} & {\bf 40.24} \\
s3\_srvr.blast.11.BUG.i.cil.c & 722.64 & {\bf 49.05} & 1515.26 & 207.09 & {\bf 49.05} \\
s3\_srvr.blast.12.BUG.i.cil.c & 620.03 & {\bf 38.66} & 555.60 & 39.28 & {\bf 38.66} \\
s3\_srvr.blast.13.BUG.i.cil.c & 831.45 & {\bf 251.56} & 1600.65 & 626.93 & {\bf 251.56} \\
s3\_srvr.blast.14.BUG.i.cil.c & 773.26 & 53.93 & 557.13 & {\bf 39.94} & {\bf 39.94} \\
s3\_srvr.blast.15.BUG.i.cil.c & $>$1800.00 & 77.51 & 530.85 & {\bf 40.19} & {\bf 40.19} \\
s3\_srvr.blast.16.BUG.i.cil.c & 973.44 & 55.97 & 558.44 & {\bf 39.54} & {\bf 39.54} \\
\hline
{\bf TOTAL (solved/time)} & {\bf ~20\,/\,10159.92} & {\bf ~22\,/\,1637.84} & {\bf ~20\,/\,12772.72} & {\bf ~20\,/\,1675.56} & {\bf ~24\,/\,2296.25}\\
\hline
\end{tabular}
\end{smallerfont}
\end{table*}

\end{document}